%% file: main.tex
\documentclass{aa}  

\usepackage{graphicx}
\usepackage{txfonts,dsfont}
\usepackage[colorlinks=true, linkcolor=blue, citecolor=blue, filecolor=blue, urlcolor=blue]{hyperref}

\usepackage{subcaption}
\usepackage{rotating}
\usepackage{multirow}
\usepackage{arydshln}
\usepackage{tikz,cancel,ulem}
\usepackage{pifont}
\usepackage{enumitem}
\usepackage{algorithm}
\usepackage{algpseudocode}
\usepackage[colorlinks=true,
    linkcolor=blue, citecolor=blue, filecolor=blue, urlcolor=blue]{hyperref} 
\usepackage{soul}

\usetikzlibrary{positioning}

\input{inputs/mycommands.tex}
\input{inputs/mycommands_specific.tex}

{\bfseries}{\itshape}

\begin{document}

\title{A plug-and-play approach with fast uncertainty quantification for weak-lensing mass mapping}

    \author{
            H. Leterme\inst{1,2, \thanks{\texttt{hubert.leterme@cea.fr}}}
        \and
            A. Tersenov\inst{2,3,\thanks{\texttt{atersenov@physics.uoc.gr}}}
        \and
            J. Fadili\inst{1}
        \and
            J.-L. Starck\inst{2,3}
    }

   \institute{Universit\'e Caen Normandie, ENSICAEN, CNRS, Normandie Univ, GREYC UMR 6072, F-14000 Caen, France
        \and
          Universit\'e Paris-Saclay, Universit\'e Paris Cit\'e, CEA, CNRS, AIM, 91191, Gif-sur-Yvette, France
         \and
             Institutes of Computer Science and Astrophysics, Foundation for Research and Technology Hellas (FORTH), Heraklion, Greece       
             }
             

 
\abstract
{
    Upcoming stage-IV surveys such as Euclid and Rubin will deliver vast amounts of high-precision data, opening new opportunities to constrain cosmological models with unprecedented accuracy. A key step in this process is the reconstruction of the dark matter distribution from noisy weak-lensing shear measurements.
}
{
    Current deep-learning-based mass-mapping methods achieve high reconstruction accuracy, but either require retraining a model for each new observed sky region (limiting practicality) or rely on slow Markov chain Monte Carlo sampling. Efficient exploitation of future survey data therefore calls for a new method that is accurate, flexible, and fast at inference. In addition, an uncertainty quantification with coverage guarantees is essential for a reliable cosmological parameter estimation.
}
{
    We introduce PnPMass, a plug-and-play approach for weak-lensing mass mapping. The algorithm produces point estimates by alternating between a gradient descent step with a carefully chosen data fidelity term and a denoising step implemented with a single deep-learning model trained on simulated data corrupted by Gaussian white noise.
    We also propose a fast sampling-free uncertainty quantification scheme based on moment networks, with calibrated error bars obtained through conformal prediction to ensure coverage guarantees. Finally, we benchmark PnPMass against model-driven and data-driven mass-mapping techniques.
}
{
    PnPMass achieves a performance close to that of the currently best deep-learning methods while offering fast inference. It converges in just a few iterations, and it requires only a single training phase, regardless of the noise covariance of the observations.
    It therefore combines flexibility, efficiency, and reconstruction accuracy while delivering tighter error bars than existing approaches, making it well suited for upcoming weak-lensing surveys.
}
{}



\maketitle

%


\section{Introduction}
\label{sec:intro}

In the coming years, next-generation observational instruments such as the Euclid space telescope and the Vera C.\@ Rubin Observatory will deliver galaxy catalogs with unprecedented precision and sky coverage. These surveys are expected to enhance our understanding of large-scale structures in the Universe and to refine existing models of its formation.

A central task in analyzing data from such surveys is mass mapping through weak gravitational lensing. The aim is to reconstruct the distribution of dark matter from noisy shear measurements, which capture the distortion of galaxy shapes caused by inhomogeneous mass distribution along the line of sight \citep{Kaiser1993}.
The resulting mass maps are then used to constrain the plausible ranges of cosmological parameters.
An accurate mass mapping is crucial for optimizing cosmological inference. Recent work by \citet{TersenovImpactWeaklensingMassmapping2025} showed that replacing the classical Kaiser-Squires method \citep{Kaiser1993} with a more advanced algorithm such as MCALens \citep{Starck2021} can significantly improve the figure of merit (FoM) in realistic weak-lensing pipelines.
Moreover, it is essential to rigorously quantify the uncertainty because errors in mass reconstruction can propagate to the estimation of cosmological parameters and might amplify discrepancies.

In recent years, deep-learning-based approaches have achieved the currently best reconstruction accuracy in weak-lensing mass mapping.
Two notable examples are DeepMass \citep{Jeffrey2020}, which directly predicts a point estimate from a noisy shear map, and DeepPosterior \citep{Remy2023}, a Bayesian method that samples from the posterior distribution of the mass map.
Despite their performances, both methods face limitations that are incompatible with upcoming stage-IV surveys such as Euclid and Rubin, where vast amounts of data are set to be delivered.
DeepMass offers fast inference because a single forward pass through the network suffices to produce a point estimate from an initial computationally inexpensive estimate derived from a noisy observation. However, as an end-to-end method, the network is sensitive to factors such as varying noise levels and masks, and it must therefore be retrained for each new observation. DeepPosterior, by contrast, is more flexible and does not require retraining, but inference is computationally expensive since multiple samples must be drawn to obtain a single point estimate.

We introduce PnPMass, a novel deep-learning-based method for weak-lensing mass mapping that like DeepMass generates accurate point estimates without posterior sampling and like DeepPosterior does not require retraining for each new observation.
Our approach builds on the plug-and-play (PnP) framework \citep{Venkatakrishnan2013}, which enables efficient reconstruction by integrating deep priors into physical models. PnP methods have become widespread in computational imaging and physical modeling \citep{Kamilov2023} and have seen growing theoretical support in recent years \citep{Ryu2019,Terris2020,Hurault2022,Ebner2024}, including convergence guarantees to a fixed point under appropriate conditions.
Our key insight is to show that with appropriate algorithmic design, particularly regarding the data-fidelity term, it is possible to train a denoiser on white Gaussian noise with a standard deviation randomly drawn from a given interval, regardless of the noise characteristics of the actual observation. This allows for a high flexibility across datasets and survey conditions while maintaining fast and accurate reconstructions.

Additionally, we propose an uncertainty quantification (UQ) method for PnPMass that does not rely on sampling. Specifically, we adapted the moment network approach by \citet{Jeffrey2020a}, originally developed for end-to-end methods such as DeepMass, to our PnP framework.
To obtain statistically calibrated uncertainty estimates, we further applied conformalized quantile regression (CQR) \citep{Romano2019} to derive per-pixel error bars with provable coverage guarantees.

Our experiments, conducted on the $\kappa$TNG simulated dataset \citep{Osato2021}, demonstrate that PnPMass achieves competitive reconstruction accuracy while remaining flexible. The model is trained only once, unlike DeepMass, and it only requires a few iterations at inference (significantly fewer than those needed by DeepPosterior for sampling). Furthermore, our UQ method, which requires only one additional iteration to compute per-pixel intervals, produces smaller error bars than DeepMass. A qualitative comparison of PnPMass with DeepMass and DeepPosterior is provided in Table~\ref{table:qualitative_comp}.

\begin{table}
    \caption{Qualitative comparison of the mass-mapping methods.}
    \label{table:qualitative_comp}
    \vspace{-5pt}
    \small
    \centering
	\setlength{\tabcolsep}{4pt}
	\renewcommand{\arraystretch}{1.1}
	\begin{tabular}{l | c c c c c}
        \hline\hline
		& Accurate & Flexible & Fast rec. & Fast UQ \\
		\hline
		DeepMass & \cmark & \xmark$^\ast$ & \cmark & \cmark \\
		DeepPosterior & \cmark & \cmark & \xmark & \xmark \\
		\hline
		\textbf{PnPMass (ours)} & \cmark & \cmark & \cmark & \cmark \\
        \hline
	\end{tabular}
    \begin{minipage}{\linewidth}
    \vspace{5pt}
    \footnotesize
    \textbf{Notes.} 
    $^\ast$Requires specific retraining for each new observation.
    \end{minipage}
\end{table}

This paper is organized as follows. In Sect.~\ref{sec:background} we review the mass-mapping problem and existing data-driven reconstruction methods. In Sect.~\ref{sec:pnp} we present the PnPMass framework. In Sect.~\ref{sec:uq} we introduce our UQ pipeline, including the method based on moment networks in Sect.~\ref{subsec:uq_mn}, and CQR-based calibration in Sect.~\ref{subsec:uq_cqr}. Sections~\ref{sec:experiments} and \ref{sec:results} detail our experimental setup and results. We conclude and discuss future directions in Sect.~\ref{sec:concl}.

\section{Background on weak-lensing mass mapping}
\label{sec:background}

We denote deterministic vectors with bold lower-case Greek or Latin letters (\eg $\convmap$), while random vectors are represented with bold sans-serif capital letters (\eg $\convmapRand$). Deterministic matrices are indicated by standard bold capital letters (\eg $\convToShear$). Furthermore, indexing is done using brackets ($\convmap[k]$,  $\convmapRand[k]$, or $\convToShear[k,\, l]$).

\subsection{The mass-mapping problem}
\label{subsec:background_mm}

This section provides a brief overview of the weak-lensing mass-mapping problem, for which a comprehensive review has been proposed by \citet{Kilbinger2015}. The objective is to recover a convergence map $\convmap \in \mathR^{\imgsize^2}$, which causes isotropic dilations of background galaxies, from a shear map $\shearmap \in \mathC^{\imgsize^2}$, which induces anisotropic stretching. Both fields are discretized over a square grid of size $\imgsize \times \imgsize$, with $\imgsize \in \mathN$, and represented as flattened one-dimensional vectors. In the weak-lensing regime, the convergence is proportional to the projected mass along the line of sight between the background galaxies and the observer \citep{Kaiser1993}. Reconstructing the convergence field is thus equivalent to estimating the projected mass distribution.

In practice, only variations in $\convmap$ around its mean value can be recovered because of the mass-sheet degeneracy. Similarly, only the variations in $\shearmap$ around its mean value contribute to the reconstruction. We therefore assume, without loss of generality, that both $\convmap$ and $\shearmap$ are zero-centered. Under this assumption, a first-order relation between shear and convergence can be expressed as
\begin{equation}
    \shearmap = \convToShear \convmap,
\label{eq:invprob0}
\end{equation}
where $\convToShear \in \mathC^{\imgsize^2 \times \imgsize^2}$ is a linear operator, defined as
\begin{equation}\label{eq:defconvtoshear}
    \convToShear := \fouriermatrHerm\fourierConvToShear\fouriermatr,
\end{equation}
where $\fouriermatr$ and its adjoint (Hermitian, or conjugate, transpose) $\fouriermatrHerm \in \mathC^{\imgsize^2 \times \imgsize^2}$ represent the two-dimensional discrete Fourier and inverse Fourier transforms, respectively. Moreover, $\fourierConvToShear \in \mathC^{\imgsize^2 \times \imgsize^2}$ is a diagonal matrix satisfying $\fourierConvToShear[0,\, 0] = 0$, and for all $(k_1,\, k_2) \in \oneto{\imgsize-1}^2 \setminus \{(0,\, 0)\}$,%
\begin{equation}\label{eq:deffourierconvtoshear}
    \fourierConvToShear\left[
        \imgsize k_1 + k_2,\, \imgsize k_1 + k_2
    \right] := \frac{(k_1 + i k_2)^2}{k_1^2 + k_2^2}.
\end{equation}
Because $\fourierConvToShear$ is zero at the zero frequency, the operator $\convToShear$ is not invertible. Its kernel is the subspace of constant vectors, whereas its image is the subspace of zero-mean vectors. The pseudo-inverse of $\convToShear$ is equal to its adjoint $\convToShearHerm$, which shares the same kernel and image. Therefore, the convergence map can be recovered via
\begin{equation}
    \convmap = \convToShearHerm \shearmap,
\label{eq:invprob0_exact}
\end{equation}
yielding an exact reconstruction when both $\shearmap$ and $\convmap$ are restricted to the subspace of zero-mean vectors.

A direct observation of the shear, however, is impossible in practice. Instead, a noisy estimator can be obtained by binning a galaxy catalog at the desired resolution and averaging the corresponding galaxy ellipticities within each pixel. Then, \eqref{eq:invprob0} becomes
\begin{equation}
    \shearmap = \mask\convToShear \convmap + \noise,
\label{eq:invprob_masked}
\end{equation}
where $\mask \in \tuple{0}{1}^{\imgsize^2 \times \imgsize^2}$ is a diagonal matrix encoding the observational mask (\eg, survey boundaries or contamination from bright stars in the foreground), and $\noise \in \mathC^{\imgsize^2}$ is a realization of proper complex Gaussian noise $\noiseRand$ with zero-mean and diagonal covariance matrix $\covmatrNoise \in \mathR^{\imgsize^2 \times \imgsize^2}$. For any pixel $k \in \smalloneto{\imgsize^2}$ where $\mask[k,\, k] = 1$, the variance is given by
\begin{equation}
    \covmatrNoise[k,\, k] = \stdval_{\ellip}^2 / \ngalperpix{k},
\label{eq:covmatrnoise}
\end{equation}
where $\stdval_{\ellip}^2 > 0$ is the variance of intrinsic ellipticities, and $\ngalperpix{k}$ is the number of galaxies observed in pixel $k$.

For pixels $k \in \smalloneto{\imgsize^2}$ such that $\mask[k,\, k] = 0$ (\ie, within masked regions), we adopted a strategy similar to that proposed by \citet{Starck2021,Remy2023}: the noise variance $\covmatrNoise[k,\, k]$ was artificially set to a very high value in order to reduce the signal-to-noise ratio (S/N) to near zero.
Then, \eqref{eq:invprob_masked} can be approximated by
\begin{equation}
    \shearmap = \convToShear \convmap + \noise.
\label{eq:invprob}
\end{equation}
In the remainder of the paper, pixels outside and within the mask are referred to as active pixels and masked pixels, respectively.

\subsection{Overview of mass-mapping methods}
\label{subsec:background_relatedwork}

\subsubsection{Classical approaches}
\label{subsubsec:background_relatedwork_classical}

The most simple method for the mass-mapping problem was proposed by \citet{Kaiser1993} and is referred to as the Kaiser-Squires (KS) solution. It consists of applying \eqref{eq:invprob0_exact} to the noisy shear map $\shearmap$, followed by Gaussian smoothing to reduce the noise.
More recent approaches have proposed to adopt a variational framework by solving minimization problems of the form
\begin{equation}
    \convmapEstimate \in \argmin_{\convmap' \in \mathR^{\imgsize^2}} \datafidelity_{\shearmap}(\convToShear\convmap') + \regfun(\convmap'),
\label{eq:minprob}
\end{equation}
where $\datafidelity_{\shearmap}(\convToShear\convmap')$ represents a data fidelity term measuring the resemblance between $\shearmap$ and $\convToShear\convmap'$, while $\regfun$ denotes a handcrafted regularization function penalizing unexpected solutions.
Notable examples include iterative Wiener filtering \citep{Bobin2012}, a forward-backward splitting algorithm with $\ell^2$ regularization corresponding to a Gaussian prior with a predefined power spectrum; the GLIMPSE2D algorithm \citep{Lanusse2016}, which enforces sparsity in a wavelet dictionary; and the MCALens algorithm \citep{Starck2021}, which models the solution as a combination of a Gaussian and a sparse component in a wavelet dictionary.

\subsubsection{Data-driven approaches}
\label{subsubsec:background_relatedwork_dl}

Deep-learning models have demonstrated significantly improved reconstruction accuracy over classical approaches by learning the distribution of convergence maps from large simulation datasets. From a Bayesian perspective, where the convergence map $\convmap$ is viewed as the realization of a random vector $\convmapRand$ with an unknown distribution, the inverse problem \eqref{eq:invprob} becomes
\begin{equation}
    \shearmapRand = \convToShear \convmapRand + \noiseRand,
\label{eq:invprob_randbayes}
\end{equation}
where $\shearmapRand$ denotes the random vector from which $\shearmap$ is drawn.

One of the deep-learning-based approaches is DeepMass \citep{Jeffrey2020}. It takes the noisy shear map $\shearmap$ as input, computes a naive estimate such as the KS or iterative Wiener solution, passes it through a deep neural network (typically, a UNet), and outputs an enhanced reconstruction $\convmapEstimate_{\deepmass}$. The model is trained on simulated pairs of convergence maps and their corresponding noisy shear maps, $\trainingsetMM$, following the forward model \eqref{eq:invprob}.

While DeepMass achieves high reconstruction accuracy, it has a major limitation: it requires retraining for each new observation $\shearmap$, as the noise covariance $\covmatrNoise$ depends on the galaxy number density in each pixel. As a result, the training data must be adapted accordingly, which reduces the flexibility of the method.

As an alternative, DeepPosterior \citep{Remy2023} generates samples from the posterior $\condpriorpdf{\cond{\convmapRand}{\shearmapRand=\shearmap}}{\cdot}$ by using a Markov chain Monte-Carlo (MCMC) algorithm. In this approach, the prior $\priorpdf{\convmapRand}{\cdot}$ is implicitly learned via neural score matching. Then, averaging over the samples yields a point estimate $\convmapEstimate_{\deepposterior}$. This method directly allows for UQ by computing sample statistics. In contrast to the previous approach, DeepPosterior only needs to be trained once, regardless of the noise properties. However, inference is slow because many samples are required to produce a single estimate.

Other methods have been proposed, but have similar limitations. \Citet{Shirasaki2019,Shirasaki2021,AoyamaDenoisingWeak2025}
introduced mass-mapping techniques based on conditional generative adversarial networks (GANs) or diffusion models.
Recently, \citet{WhitneyGenerativeModellingMassmapping2025} developed a fast sampling method based on conditional Wasserstein GANs, which significantly reduced the generation time per sample compared to DeepPosterior.
Despite their competitive performances, the above methods are limited in the same way as DeepMass: the training phase is specific to a given observation, and the methods therefore lack flexibility. More precisely, given a newly observed sky area or another galaxy survey, these models must be retrained relatively to the new configuration, with a specific mask $\mask$ and per-pixel noise level $\covmatrNoise$. Our approach specifically addresses this issue.

\section{Proposed approach based on plug-and-play}
\label{sec:pnp}

In this section, we introduce PnPMass, a PnP-based algorithm specifically tailored for the mass-mapping problem.
We cast the problem into a real-valued formulation, where we define%
\begin{equation}
    \tilde\shearmap := \begin{pmatrix}
        \Real\shearmap \\ \Imag\shearmap
    \end{pmatrix} \in \mathR^{2\imgsize^2},
    \qquad
    \tilde\noise := \begin{pmatrix}
        \Real\noise \\ \Imag\noise
    \end{pmatrix} \in \mathR^{2\imgsize^2},
\end{equation}
\begin{equation}\label{eq:deftildeAtildeSigma}
    \tilde\convToShear := \begin{pmatrix}
        \Real\convToShear \\ \Imag\convToShear
    \end{pmatrix} \in \mathR^{2\imgsize^2 \times \imgsize^2},
    \quad
    \tilde\covmatrNoise := \frac12\begin{pmatrix}
        \covmatrNoise & \bzero \\
        \bzero & \covmatrNoise
    \end{pmatrix} \in \mathR^{2\imgsize^2 \times 2\imgsize^2}.
\end{equation}
Then, \eqref{eq:invprob} becomes
\begin{equation}
    \tilde\shearmap = \tilde\convToShear \convmap + \tilde\noise,
\label{eq:invprob_lifted}
\end{equation}
where the noise $\tilde\noise$ is a realization of $\tilde\noiseRand \sim \MN(\bzero,\, \tilde\covmatrNoise)$.

In Sect.~\ref{subsec:pnp_fixedpoint} we introduce a fixed-point iteration composed of a forward step, involving a linear operator $\tilde\genericop$ and a step size $\stepsize > 0$, followed by a backward step using a deep denoiser $\deepmodelTrained$.
Sect.~\ref{subsec:pnp_smallresidual} establishes properties on $\deepmodelTrained$ that enable an accurate recovery of the convergence map $\convmap$ from the noisy shear input $\shearmap$.
In Sect.~\ref{subsec:pnp_operator} we show that a suitable choice of $\tilde\genericop$ allows the denoiser to be trained on white Gaussian noise, regardless of the noise covariance matrix $\tilde\covmatrNoise$. Sect.~\ref{subsec:pnp_stepsize} then derives bounds on $\stepsize$ to guarantee convergence to a fixed point.
Additionally, Sect.~\ref{subsec:pnp_denoiserprops} details properties required for the denoiser, including fixed-point conditions and noise-level awareness.
Finally, in Sect.~\ref{subsec:pnp_respnp}, we present a variant of PnPMass that incorporates prior knowledge of the Gaussian structure of convergence maps outside peak regions dominated by large matter aggregates.

\subsection{Fixed-point iteration}
\label{subsec:pnp_fixedpoint}

Given an observation $\tilde\shearmap \in \mathR^{2\imgsize^2}$ and a step-size $\stepsize > 0$, we introduce a forward operator $\forwardop{\tilde\shearmap}(\cdot,\, \stepsize): \mathR^{\imgsize^2} \to \mathR^{\imgsize^2}$, defined for any input $\convmap'$ by
\begin{equation}
    \forwardop{\tilde\shearmap}(\convmap',\, \stepsize) := \convmap' + \stepsize\tilde\genericop(\tilde\shearmap - \tilde\convToShear\convmap'),
\label{eq:forwardop}
\end{equation}
where $\tilde\genericop \in \mathR^{\imgsize^2 \times 2\imgsize^2}$ denotes a linear operator to be defined. We also consider a deep-learning-based denoiser $\deepmodelTrained: \mathR^{\imgsize^2} \to \mathR^{\imgsize^2}$ with trained parameters $\hat\paramDeepmodel$.
The PnP method, inspired by standard operator splitting methods in optimization theory for solving inverse problems (see \citealp{Starck2015} and Appendix~\ref{appendix:pnp_fb}) consists of iterating the operator $\fixedpointiter{\tilde\shearmap}(\cdot,\, \stepsize)$ taken as the composition of the forward operator and the trained denoiser,
\begin{equation}
    \fixedpointiter{\tilde\shearmap}(\cdot,\, \stepsize) := \deepmodelTrained \circ \forwardop{\tilde\shearmap}(\cdot,\, \stepsize).
\label{eq:fixpointoperator}
\end{equation}

Toward studying the properties of this operator, we state our main assumptions below.
\begin{enumerate}[label=(H.\arabic*),leftmargin=*]
    \item \label{assum:1}
    $\deepmodelTrained: \orthcomplBA \to \orthcomplBA$; 
    \item \label{assum:2}
     $\deepmodelTrained$ is non-expansive on $\orthcomplBA$, \ie, $\lipschitzcst$-Lipschitz continuous on $\orthcomplBA$ with $\lipschitzcst \leq 1$; 
    \item \label{assum:3}
    $\tilde\genericop\tilde\convToShear$ is symmetric positive semidefinite;
    \item \label{assum:3'}
    $\image \tilde\genericop \subset \orthcomplBA$;
    \item \label{assum:4}
    Let $\lambda_{\min}$ and $\lambda_{\max}=\smallnorm{\tilde\genericop\tilde\convToShear}$ denote the minimum and maximum nonzero eigenvalues of $\tilde\genericop\tilde\convToShear$, respectively. We also introduce
    \begin{equation}
        \rho^\star := \frac{\lambda_{\max}-\lambda_{\min}}{\lambda_{\max}+\lambda_{\min}}.
    \label{eq:lowerbound_rho}
    \end{equation}
    For fixed $\rho \in \intervalexclr{\rho^\star}{1}$, the step size $\stepsize$ satisfies
    \begin{equation}
        \frac{1 - \rho}{\lambda_{\min}} \leq \stepsize \leq \frac{1 + \rho}{\lambda_{\max}} .
    \label{eq:stepsize}
    \end{equation}   

\end{enumerate}

Under the above assumptions, the operator $\fixedpointiter{\tilde\shearmap}(\cdot,\, \stepsize)$ enjoys the following properties:
\begin{proposition}
    \label{prop:fixedpoint}
    Under assumptions \ref{assum:1}--\ref{assum:4}, the operator $\fixedpointiter{\tilde\shearmap}(\cdot,\, \stepsize)$ admits a unique fixed point $\convmapEstimate \in \orthcomplBA$.
    Moreover, the sequence of iterates $(\convmap^{(k)})_{k \in \mathN}$ defined by
    \begin{equation}
        \convmap^{(k+1)} = \fixedpointiter{\tilde\shearmap}(\convmap^{(k)},\, \stepsize),
    \label{eq:fixedpointseries}
    \end{equation}
    starting at $\convmap^{(0)} \in \orthcomplBA$, converges linearly to $\convmapEstimate$ with a convergence rate $\rho$.
\end{proposition}
\begin{proof}
    See Appendix~\ref{appendix:proof_fixedpoint}.
\end{proof}

\subsection{Which noise should be used to train the denoiser ?}
\label{subsec:pnp_smallresidual}

Under the above assumptions, $\convmapEstimate$ obeys the fixed-point equation
\begin{equation}
\convmapEstimate =  \deepmodelTrained\Bigl(\convmapEstimate + \stepsize\tilde\genericop(\tilde\shearmap - \tilde\convToShear\convmapEstimate)\Bigr).
\label{eq:fixedpoint}
\end{equation}
Assessing whether the fixed point $\convmapEstimate$ lies sufficiently close to the true convergence map $\convmap$ is a challenging question that falls beyond the scope of this paper. A mathematical framework is nevertheless sketched in Appendix~\ref{appendix:convprops}, thereby constraining the noise level at which the denoiser $\deepmodelTrained$ must be trained.

Our intuition can be summarized as follows. Assuming the residual $\residualEstimate := \convmapEstimate - \convmap$ is small enough, we obtain by applying the forward step to the fixed point $\convmapEstimate$ the approximation
\begin{align}
    \convmapEstimate_0 := \forward_{\tilde\shearmap}(\convmapEstimate,\ \stepsize)
        &= \convmap + \residualEstimate + \stepsize\tilde\genericop\bigl(
            -\tilde\convToShear\residualEstimate + \tilde\noise
        \bigr)
\label{eq:fixedpointcondition0}\\
        &\approx \convmap + \stepsize\tilde\genericop\tilde\noise.
\label{eq:fixedpointcondition}
\end{align}
Therefore, application of the operator $\fixedpointiter{\tilde\shearmap}$ introduced in \eqref{eq:fixpointoperator} to the fixed point $\convmapEstimate$ yields
\begin{align}
    \convmapEstimate
        &= \deepmodelTrained\bigl(
            \convmapEstimate_0
        \bigr) \approx \deepmodelTrained\bigl(
            \convmap + \stepsize\tilde\genericop\tilde\noise
        \bigr).
\end{align}
Thus, $\deepmodelTrained$ acts as a denoiser on the true convergence map $\convmap$ corrupted by a Gaussian noise $\noise_0 := \stepsize\tilde\genericop\tilde\noise \in \mathR^{\imgsize^2}$, which is an outcome of
\begin{equation}
    \noiseRand_0 := \stepsize\tilde\genericop\tilde\noiseRand \sim \MN(\bzero,\, \covmatrNoise_0),\qwith \covmatrNoise_0 := \stepsize^2\tilde\genericop\tilde\covmatrNoise\tilde\genericop^\top.
\label{eq:pnpnoise}
\end{equation}

In practice, we trained the denoiser by minimizing the empirical quadratic risk,
\begin{equation}\label{eq:traindenoiser}
\hat\paramDeepmodel \in \argmin_{\paramDeepmodel} \frac{1}{n} \sum_{i=1}^n \norm{\deepmodelParameterized\Bigl(\convmap_i+\noise_{0,\, i}\Bigr) - \convmap_i}^2 ,
\end{equation}
using a dataset of $\sizeTrainingset$ training samples. The ground-truth convergence maps $\convmap_i$ were generated from simulations based on the $\Lambda$CDM model, and the noisy inputs were computed from $\convmap_i$ by adding noise $\noise_{0,\, i}$ drawn according to \eqref{eq:pnpnoise}.

\subsection{Choice of the operator $\tilde\genericop$}
\label{subsec:pnp_operator}

One of the main objectives of this paper is to implement a mass-mapping method whose training phase is independent of the training noise covariance matrix $\tilde\covmatrNoise$. Clearly, given our discussion in Sect.~\ref{subsec:pnp_smallresidual}, the covariance matrix $\covmatrNoise_0$ in \eqref{eq:pnpnoise} must be independent of $\tilde\covmatrNoise$.
To this end, we used the degree of freedom offered by the choice of $\tilde\genericop$. More precisely, we propose to make the following choice:
\begin{equation}
    \tilde\genericop := \tildeConvToShearTransp\tilde\covmatrNoise^{-1/2},
\label{eq:genericop}
\end{equation}
in which case~\eqref{eq:pnpnoise} and \eqref{eq:forwardop} become
\begin{equation}
	\covmatrNoise_0 = \tau^2\tildeConvToShearTransp\tilde\convToShear \qand
    \forwardop{\tilde\shearmap} (\convmap',\, \stepsize) := \convmap' + \stepsize\tildeConvToShearTransp\tilde\covmatrNoise^{-1/2}(\tilde\shearmap - \tilde\convToShear\convmap').
\label{eq:forwardop2}
\end{equation}
We caution here that strictly speaking, $\covmatrNoise_0$ is not a covariance matrix as it is only semidefinite positive with $\ker \covmatrNoise_0 = \ker \convToShear$, the subspace of constant vectors. This is not a concern, however, as we are only interested in random perturbations along $\orthcomplBA = (\ker \convToShear)^\perp$, where $\covmatrNoise_0$ is symmetric positive definite.

In Appendix~\ref{appendix:orthogonalcompl} we show that with the choice \eqref{eq:genericop}, $\tilde\genericop\tilde\convToShear = \tildeConvToShearTransp\tilde\covmatrNoise^{-1/2}\tilde\convToShear$ verifies assumptions \ref{assum:3} and \ref{assum:3'}. We even have the equality $\image \tilde\genericop = \orthcomplBA$.

 \eqref{eq:forwardop2} shows that the residual $\tilde\shearmap - \tilde\convToShear\convmap'$ is multiplied by $\tilde\covmatrNoise^{-1/2}$, which corresponds to a noise-whitening step. This is in a stark contrast with the choice $\tilde\genericop \!:=\! \tildeConvToShearTransp\tilde\covmatrNoise^{-1}$, which yields $\forwardop{\tilde\shearmap} (\convmap',\, \stepsize) := \convmap' + \stepsize\tildeConvToShearTransp\tilde\covmatrNoise^{-1}(\tilde\shearmap - \tilde\convToShear\convmap')$. This alternative choice is the one dictated by the standard variational (or Bayesian MAP) perspective with a data fidelity
$\tfrac{1}{2}\bignorm{\tilde\covmatrNoise^{-1/2}(\tilde\shearmap - \tilde\convToShear\convmap')}^2$
as negative log-likelihood, whose gradient is $\tildeConvToShearTransp\tilde\covmatrNoise^{-1}(\tilde\shearmap - \tilde\convToShear\convmap')$.
This clearly shows that the widely adopted optimization or Bayesian perspective on PnP is not necessarily the most appropriate and might even be misleading in the PnP context (see also the discussion in Appendix~\ref{appendix:pnp_fb}). The fixed-point perspective, on the other hand, offers more insight and flexibility.%

In Appendix~\ref{appendix:almostidentity} we show that
\begin{equation}
\tildeConvToShearTransp\tilde\convToShear = \BBI - \frac1{\imgsize^2},
\label{eq:almostidentity}
\end{equation}
(the constant $1 / \imgsize^2$ is subtracted from all matrix elements). In practice, the image size $\imgsize$ is sufficiently large to neglect the non-diagonal elements. Therefore, from \eqref{eq:forwardop2} and \eqref{eq:almostidentity}, it is valid to take the approximation
\begin{equation}
    \covmatrNoise_0 \approx \stepsize^2 \BBI.
\label{eq:covmatrwhitenoise}
\end{equation}
In turn, the denoiser $\deepmodelParameterized$ is trained on zero-mean white Gaussian noise with a standard deviation equal to the step size $\stepsize$. In Sect.~\ref{subsec:pnp_stepsize} we discuss the selection of an appropriate range for $\stepsize$.

\subsection{Step-size selection}
\label{subsec:pnp_stepsize}

As discussed in Sect.~\ref{subsec:pnp_fixedpoint}, the existence and uniqueness of a fixed point as well as convergence of the PnP iteration \eqref{eq:fixedpointseries} to this fixed point require that $\stepsize$ satisfies \ref{assum:4}. When we substitute the expression for $\tilde\genericop$ from~\eqref{eq:genericop}, condition \eqref{eq:stepsize} specializes to
\begin{equation}
 \frac{1 - \rho}{\lambda_{\min}} \leq \stepsize \leq \frac{1 + \rho}{\lambda_{\max}} , \qwhere \lambda_{\max} = \bignorm{\tildeConvToShearTransp\tilde\covmatrNoise^{-1/2}\tilde\convToShear} .
\label{eq:stepsize2}
\end{equation}
It is well known that $\lambda_{\max}$ can be accurately computed via the power iteration method. In Appendix~\ref{appendix:powitlambdamin} we also discuss how the power iteration can be used to efficiently compute $\lambda_{\min}$.  

A simple estimate of the upper bound for $\stepsize$ can also be obtained via an upper bound of $\lambda_{\max}$. We observe that
\begin{align*}
    \bignorm{\tildeConvToShearTransp\tilde\covmatrNoise^{-1/2}\tilde\convToShear}
        &\leq \bignorm{\tildeConvToShearTransp\tilde\convToShear} \bignorm{\tilde\covmatrNoise^{-1/2}} \\
        &\leq \bignorm{\tilde\covmatrNoise^{-1/2}} = 1 / \min_k \tilde\stdval_k,
\end{align*}
where we used \eqref{eq:almostidentity}, and $\tilde\stdval_k := \tilde\covmatrNoise[k,\, k]^{1/2}$ denotes the noise standard deviation at pixel $k$. Therefore,
\begin{equation}
    \stepsize \leq (1+\rho) \min_k \tilde\stdval_k,
\label{eq:stepsize3}
\end{equation}
is a sufficient condition for the right-hand side of~\eqref{eq:stepsize2} to hold. Thus, the choice of the largest step size is  essentially governed by the lowest noise level in the input shear map, that is, the pixel with the largest number of observed galaxies.

\subsection{Denoiser training}
\label{subsec:pnp_denoiserprops}

\subsubsection{Ensuring convergence to a fixed point}

The denoiser ought to comply with \ref{assum:1} and \ref{assum:2} for Proposition~1 to hold.
As discussed in Sect.~\ref{subsec:pnp_operator} (see also Appendix~\ref{appendix:orthogonalcompl}), with the choice of $\tilde\genericop$ given in~\eqref{eq:genericop}, we have
\begin{equation}
    \orthcomplBA := \ker (\tilde\convToShear)^\perp ,
\label{eq:orthogonalcompl}
\end{equation}
that is, the subspace of zero-mean convergence maps.
In turn, \ref{assum:1} tells us that the denoiser must not introduce any bias, that is, the denoised version of a zero-mean map must also be zero-mean. This is a natural and expected property. It is consistent with the mass-sheet degeneracy property discussed in Sect.~\ref{subsec:background_mm} and can be enforced by explicitly appending a mean-centering layer at the output of the denoiser.

The non-expansiveness assumption \ref{assum:2} can be encouraged during training by introducing a regularization term that controls the spectral norm of the denoiser Jacobian \citep{Pesquet2021}. However, this strategy is computationally very expensive, and we simply chose to ignore this constraint in practice. A theoretical understanding of the training dynamics, including a potential implicit bias toward non-expansive denoisers, would be valuable. This avenue is beyond the scope of this work and is left for future research.

\subsubsection{Noise-level-aware denoiser}
\label{subsubsec:pnp_denoiserprops_noiseaware}

According to the discussions in Sects.~\ref{subsec:pnp_smallresidual}, \ref{subsec:pnp_operator} and \ref{subsec:pnp_stepsize}, the denoiser $\deepmodelTrained$ was trained by solving \eqref{eq:traindenoiser}, where $\noise_{0,\, i}$ are sampled from a zero-mean white Gaussian distribution with variance $\stepsize^2$. The bound \eqref{eq:stepsize3} shows that the noise level $\stepsize$ used for training depends on the noise covariance matrix $\tilde\covmatrNoise$ of the mass-mapping problem. At this stage, the training phase therefore remains dependent on $\tilde\covmatrNoise$. However, this dependence now reduces to a single scalar $\stepsize$, whereas the full (diagonal) covariance matrix was required for training DeepMass.

We aim to push this one step forward by making the training phase completely independent of the covariance matrix. One simple strategy would be to select a low value for $\stepsize$ to cover a broad range of survey conditions. However, as our experiments confirm, the choice of $\stepsize$ strongly affects the reconstruction quality.

We therefore adopted a more flexible approach in which the model was enabled to adapt to different noise conditions during inference. To this end, the network was trained on a dataset $(\convmap_{0,\, i},\, \convmap_i)_{i=1}^{\sizeTrainingset}$, where each input $\convmap_{0,\, i}$ was obtained by corrupting the ground truth $\convmap_i$ by noise $\noise_{0,\, i}$ drawn from $\MN(\bzero,\, \stdval_i^2)$, with $\stdval_i$ being randomly selected among a wide range of noise levels. 
For higher reconstruction accuracy, we used a noise-level-aware model, following an approach initially proposed by \citet{ZhangPlugandPlayImage2022}. In this context, the denoiser network takes the noise standard deviation $\sigma$ as an additional entry, that is, $\deepmodelTrained(\convmap',\, \stdval)$, where $\stdval > 0$ denotes the noise standard deviation at which the model is expected to operate.

During inference with PnPMass, according to~\eqref{eq:covmatrwhitenoise}, the intermediate outputs are propagated though the network with noise standard deviation $\stepsize$. Then, the fixed-point iteration introduced in~\eqref{eq:fixpointoperator} now reads
\begin{equation}
    \fixedpointiter{\tilde\shearmap}(\cdot,\, \stepsize) := \deepmodelTrained(\cdot,\, \stepsize) \circ \forwardop{\tilde\shearmap}(\cdot,\, \stepsize) ,
\label{eq:fixpointoperator_noiseaware}
\end{equation}
where now the forward and denoising steps both depend on the step size $\stepsize$. This is again reminiscent of the forward-backward splitting (FBS) algorithm, as discussed in Appendix~\ref{appendix:pnp_fb}.

\subsubsection{Training on multiple cosmologies}

Generating training data via simulation requires selecting a set of cosmological parameters, which implicitly shapes the prior distribution from which convergence map samples are drawn. If the chosen parameters diverge significantly from the true unknown cosmology that is to be inferred, the trained denoiser may struggle to generalize and thus fail to produce accurate reconstructions. One way to overcome this is to generate training data across a distribution of plausible cosmologies, enhancing the network robustness to distributional shifts.

\subsection{PnPMass on non-Gaussian residuals}
\label{subsec:pnp_respnp}

In the spirit of DeepMass \citep{Jeffrey2020}, the PnPMass algorithm can be enhanced by considering a priori knowledge about the underlying physics and data structure. To do this, we used a hypothesis made by \citet{Starck2021} for MCALens: convergence maps can be considered the superposition of a Gaussian and a non-Gaussian field (hereafter referred to as a residual field),
\begin{equation}
    \convmap = \convmapGauss + \convmapNongauss,
\end{equation}
where $\convmapGauss$, the Gaussian component, can be estimated from the noisy shear map $\shearmap$ using an iterative Wiener filtering method (see Sect.~\ref{subsubsec:background_relatedwork_classical}), assuming the power spectrum $\powerspectrum$ is known, coined $\iterativewiener{\powerspectrum}$,
\begin{equation}
    \convmapGauss := \iterativewiener{\powerspectrum}\!(\tilde\shearmap).
\end{equation}
Then, the inverse problem \eqref{eq:invprob_lifted} can be rewritten as
\begin{equation}
    \shearmapLiftedNongauss = \tilde\convToShear \convmapNongauss + \tilde\noise,
    \label{eq:invprob_res}
\end{equation}
where $\shearmapLiftedNongauss$ denotes the residual shear map,
\begin{equation}
    \shearmapLiftedNongauss := \tilde\shearmap - \tilde\convToShear\convmapGauss.
\end{equation}
Then, PnPMass can be used to estimate $\convmapNongauss$.

In order for this approach to remain consistent with the previous theoretical framework, the denoiser must be trained on a suitable dataset. Specifically, we considered the training set $(\convmap_{0,\, i},\, \convmap_i)_{i=1}^{\sizeTrainingset}$ introduced in Sect.~\ref{subsec:pnp_smallresidual}. Then, for each $i \in \oneto{\sizeTrainingset}$, we computed the Gaussian component $\convmapGauss_i$ from the noisy input $\convmap_{0,\, i}$ using a Wiener filter with a power spectrum $\powerspectrum$. Finally, we obtained the residual training set $\bigl(\convmapNongauss_{0,\, i},\, \convmapNongauss_i\bigr)_{i=1}^{\sizeTrainingset}$ by subtracting $\convmapGauss_i$ from the noisy input and the ground truth.

\section{Uncertainty quantification}
\label{sec:uq}

In Sect.~\ref{subsec:uq_mn} we introduced a fast UQ method for PnPMass. This approach is based on moment networks that were originally proposed by \citet{Jeffrey2020a} to quantify uncertainty in end-to-end deep-learning models such as DeepMass.

To the best of our knowledge, this is the first UQ method within the PnP framework that does not rely on ensembles or posterior sampling, both of which can be computationally demanding.
Ensemble-based approaches estimate the uncertainty by training and evaluating multiple networks and by capturing variability across the model parameters \citep{Shi2021, TerrisAIRIPlugandplayAlgorithm2025}. Posterior sampling methods, in contrast, are Bayesian UQ methods that seek to approximate the posterior distribution of the reconstructed images based on some implicit or explicit priors. For example, \citet{LaumontBayesianImaging2022} combined implicit PnP priors with Langevin diffusion to draw posterior samples. Alternatively, \citet{Ekmekci2021} proposed a Bayesian extension of the PnP framework by placing a prior over the denoiser parameters and propagating the uncertainty through the reconstruction pipeline.

A sampling-free method was introduced by \citet{PostelsSamplingFreeEpistemic2019} to approximate uncertainty in Monte Carlo dropout networks. This approach could in principle be adapted to the Bayesian PnP framework proposed by \citet{Ekmekci2021} as a way to reduce computational cost. This line of research targets uncertainties in the model design and parameters (commonly referred to as epistemic uncertainty) and is left for future work. In contrast, we assumed a perfect model design and training, and we instead quantified irreducible uncertainty arising from data noise, known as aleatoric uncertainty. We refer to \citet{AbdarReviewUncertainty2021} for a broader discussion of the distinction between epistemic and aleatoric uncertainty.

Next, in Sect.~\ref{subsec:uq_cqr}, we describe the CQR algorithm \citep{Romano2019}, which was adapted for mass mapping in a previous work \citep{LetermeDistributionfreeUncertainty2025}. This method serves as a postprocessing step that can be applied to any reconstruction method with an initial uncertainty estimate, providing per-pixel non-asymptotic marginal coverage guarantees without relying on prior assumptions about the data distribution.

\subsection{Variance estimation with moment networks}
\label{subsec:uq_mn}

In Sect.~\ref{subsubsec:uq_mn_deepmass} we adapt the main principles of posterior variance estimation using moment networks, as described by \citet{Jeffrey2020a}, to DeepMass specifically. Then, in Sect.~\ref{subsubsec:uq_mn_pnp}, we extend this approach to the PnP framework and provide theoretical insights supporting its use. Finally, in Sect.~\ref{subsubsec:uq_mn_errbars}, we explain how the error bars are obtained from the estimated posterior variance.

\subsubsection{Moment networks for DeepMass}
\label{subsubsec:uq_mn_deepmass}

Given a simulated dataset $\trainingsetMMLifted$, where the pairs $(\tilde\shearmap_i,\, \convmap_i)$ are independent and identically-distributed samples of the ground-truth convergence maps and observed noisy shear maps, DeepMass attempts to solve the mass-mapping problem by learning a deep neural network $\deepmassmodelParameterized$ on $\trainingsetMMLifted$ according to%
\begin{equation}\label{eq:traindeepmass}
\hat\paramDeepmassmodel \in \argmin_{\paramDeepmassmodel} \frac{1}{n} \sum_{i=1}^n \norm{\deepmassmodelParameterized(\tilde\shearmap_i) - \convmap_i}^2 .
\end{equation}
As $n \to \infty$, \eqref{eq:traindeepmass} is expected to approach a minimizer of the population risk, that is,
\begin{equation}\label{eq:traindeepmasspop}
\hat\paramDeepmassmodel \in \argmin_{\paramDeepmassmodel} \Expval\Bigl[ \norm{\deepmassmodelParameterized(\tilde\shearmapRand) - \convmapRand}^2\Bigr] ,
\end{equation}
where the expectation is to be understood with respect to the joint distribution of $(\convmapRand,\, \tilde\shearmapRand)$.

DeepMass adopts a Bayesian perspective and attempts to interpret \eqref{eq:traindeepmasspop} as the solution to
\begin{equation}
\min_{\deepmodel \, \text{is $\mu_{\cond{\tilde\shearmapRand}{\convmapRand}}$-measurable}} \Expval\Bigl[\norm{
            \deepmodel(\tilde\shearmapRand) - \convmapRand
        }^2\Bigr],
\label{eq:minmse}
\end{equation}
where $\mu_{\cond{\tilde\shearmapRand}{\convmapRand}}$ is the conditional distribution of $\tilde\shearmapRand$ given $\convmapRand$. The closed-form solution to \eqref{eq:minmse} is known as the posterior conditional mean estimator,
\begin{equation}
    \widehat\deepmodel(\tilde\shearmap) = \bigcondexpval{\convmapRand}{\tilde\shearmapRand=\tilde\shearmap}.
\label{eq:pcm}
\end{equation}
The key assumption in DeepMass is that the following approximation holds:
\begin{equation}
    \widehat\deepmodel(\tilde\shearmap) \approx \deepmassmodelTrained(\tilde\shearmap) .
\label{eq:approx_pmean}
\end{equation}
Intuitively, this can make sense if the architecture of the neural network $\deepmassmodelParameterized$, seen as a parameterized class of functions, is expressive enough, in such a way that the two optimization problems \eqref{eq:traindeepmasspop} and \eqref{eq:minmse} are close.

This reasoning can be extended to estimate higher-order moments, including the posterior conditional variance per pixel,
\begin{equation}
    \bigcondvar{\convmapRand}{\tilde\shearmapRand=\tilde\shearmap} := \Bigcondexpval{
        \Bigl(
            \convmapRand - \bigcondexpval{\convmapRand}{\tilde\shearmapRand=\tilde\shearmap}
        \Bigr)^2}{\tilde\shearmapRand=\tilde\shearmap} ,
    \label{eq:pvar}
\end{equation}
which is known to be the solution to
\begin{equation}
\min_{\deepmodelbisvar \, \text{is $\mu_{\cond{\tilde\shearmapRand}{\convmapRand}}$-measurable}} \Expval\left[\norm{
            \deepmodelbisvar(\tilde\shearmapRand) - \bigl(\convmapRand - \widehat\deepmodel(\tilde\shearmapRand\bigr)^2
            }^2\right].
\label{eq:minmsevar}
\end{equation}
To this end, we considered a new neural network model $\deepmodelbisParameterized$ trained on a dataset $(\tilde\shearmap_i,\, \varmap_i)_{i=1}^{\sizeTrainingset}$ by solving
\begin{equation}\label{eq:traindeepmassvar}
\hat\paramDeepmodelbis \in \argmin_{\paramDeepmodelbis} \frac{1}{n} \sum_{i=1}^n \norm{\deepmodelbisParameterized(\tilde\shearmap_i) - \varmap_i}^2 ,
\end{equation}
where we defined the target $\varmap_i := \left(\convmap_i - \deepmassmodelTrained(\tilde\shearmap_i)\right)^2$. Therefore, similarly to what we have argued for \eqref{eq:traindeepmasspop}, \eqref{eq:traindeepmassvar} is an empirical version of
\begin{equation}\label{eq:traindeepmassvarpop}
\hat\paramDeepmodelbisvar \in \argmin_{\paramDeepmodelbisvar} \Expval\left[ \norm{\deepmodelbisvarParameterized(\tilde\shearmapRand) - \left(
        \convmapRand - \deepmassmodelTrained(\tilde\shearmapRand)\right)^2}^2\right] ,
\end{equation}
Comparing \eqref{eq:minmsevar} and \eqref{eq:traindeepmassvarpop}, combining \eqref{eq:approx_pmean} and \eqref{eq:pvar}, and arguing as above, we can reasonably assume that
\begin{equation}
    \bigcondvar{\convmapRand}{\tilde\shearmapRand=\tilde\shearmap} \approx \deepmodelbisTrained(\tilde\shearmap) .
\end{equation}
The corresponding model is referred to as an order-2 moment network.

\subsubsection{Adaptation to PnPMass}
\label{subsubsec:uq_mn_pnp}

The moment network approach is not directly applicable to the PnP framework for several reasons. First, PnP, which is iterative, does not necessarily have a Bayesian interpretation, let alone a posterior conditional mean (see Sect.~\ref{subsec:pnp_operator} and Appendix~\ref{appendix:pnp_fb}). We also recall that we did not wish to train any network on a dataset that depends on the noise covariance matrix $\tilde\covmatrNoise$ for flexibility reasons. However, we describe below that within the setting presented in Sect.~\ref{sec:pnp}, the posterior conditional variance can be estimated with an order-2 moment network trained on zero-mean white Gaussian noise.

Let $\deepmodelTrained$ be the denoiser network that has been trained by solving \eqref{eq:traindenoiser} on a training dataset $(\convmap_{0,\, i},\, \convmap_i)_{i=1}^{\sizeTrainingset}$, where we have denoted $\convmap_{0,\, i} := \convmap_i + \noise_{0,\, i}$, with $\noise_{0,\, i}$ being sampled from $\MN(\bzero,\stepsize_i^2\BBI)$ conditionally on $\stepsize_i$ (see \eqref{eq:pnpnoise} and \eqref{eq:covmatrwhitenoise}).
The noise standard deviations $\stepsize_i$ are samples from the uniform distribution on an interval $\interval{\stepsize_{\min}}{\stepsize_{\max}}$ containing the bounds in \eqref{eq:stepsize2}. As discussed in Sect.~\ref{subsubsec:pnp_denoiserprops_noiseaware}, $\stepsize$ is also an input of the denoiser network. 

We now consider the random vector
\begin{equation}
    \convmapRand_0 := \convmapRand + \noiseRand_0,
\label{eq:convmaprand0}
\end{equation}
where $\noiseRand_0$ has been introduced in \eqref{eq:pnpnoise}. In the context in which we consider the step size $\stepsize$ to be an outcome of $\stepsizeRand \sim \MU(\interval{\stepsize_{\min}}{\stepsize_{\max}})$, $\noiseRand_0$ is now a mixture of zero-mean Gaussians with uniform mixing over $\stepsizeRand$. The noisy inputs $\convmap_{0,\, i}$ from the training set are then assumed to be drawn independently from the same distribution as $\convmapRand_0$.
With a similar reasoning as in Sect.~\ref{subsubsec:uq_mn_deepmass}, we can infer that given $\stepsize$ and a noisy input $\convmap_0$,%
\begin{equation}
    \deepmodelTrained(\convmap_0,\, \stepsize) \approx \Bigcondexpval{\convmapRand}{\convmapRand_0=\convmap_0,\, \stepsizeRand=\stepsize} .
\label{eq:approx_pmean_denoiser}
\end{equation}
For the posterior conditional variance, we consider a model $\deepmodelbisTrained$ trained on $(\convmap_{0,\, i},\, \varmap_i)_{i=1}^{\sizeTrainingset}$, with
\begin{equation}
    \varmap_i := \left(
        \convmap_i - \deepmodelTrained(\convmap_{0,\, i},\, \stepsize_i)
    \right)^2,
\label{eq:groundtruth_order2}
\end{equation}
and again following Sect.~\ref{subsubsec:uq_mn_deepmass}, we obtain
\begin{equation}
\deepmodelbisTrained(\convmap_0,\, \stepsize) = \Bigcondvar{\convmapRand}{\convmapRand_0=\convmap_0,\, \stepsizeRand=\stepsize}.
\label{eq:approx_pvar_denoiser}
\end{equation}

Now, we argue that after reaching a sufficient number of iterations, the outputs of $\deepmodelTrained$ and $\deepmodelbisTrained$ approximate the posterior conditional mean and variance, respectively, given a noisy observation $\tilde\shearmapRand = \tilde\shearmap$.
To this end, let $\convmapEstimate_0 := \forwardop{\tilde\shearmap}(\convmapEstimate,\ \stepsize)$
denote the output of the forward step applied to the fixed point, satisfying \eqref{eq:fixedpointcondition0}.
In Sect.~\ref{subsec:pnp_fixedpoint}  we informally established conditions on the training noise to favor small residuals $\residualEstimate$. Therefore, we consider \eqref{eq:fixedpointcondition} as a reasonable approximation,
\begin{equation}
    \convmapEstimate_0 \approx \convmap + \noise_0,\, \qwith \noise_0 := \stepsize\tilde\genericop\tilde\noise \in \mathR^{\imgsize^2}.
\label{eq:fixedpointcondition_reminder}
\end{equation}
For a number of iterations $\niter \geq 1$, denote $\fixedpointiter{\tilde\shearmap}^{\niter}(\cdot,\stepsize)$ the $\niter$-th composition of the PnPMass fixed-point operator $\fixedpointiter{\tilde\shearmap}^{\niter}(\cdot,\stepsize)$. Clearly, the PnPMass iteration \eqref{eq:fixedpointseries} also reads
\begin{equation}
    \convmap^{(\niter)} = \fixedpointiter{\tilde\shearmap}^{\niter}(\convmap^{(0)},\stepsize) .
\end{equation}
We denote $\pnpmassnoisy(\tilde\shearmap,\, \stepsize) := \forwardop{\tilde\shearmap}(
    \convmap^{(\niter)},\ \stepsize
)$ the output of the $(\niter + 1)$-th forward step.
For $\niter$ large enough, we have from Proposition~1 that $\convmap^{(\niter)} \approx \convmapEstimate$, which combined with \eqref{eq:fixedpointcondition0} and \eqref{eq:fixedpointcondition} yields
\begin{equation}\label{eq:fixedpointcondition_rand}
\pnpmassnoisy(\tilde\shearmap,\, \stepsize) \approx \forwardop{\tilde\shearmap}(\convmapEstimate,\ \stepsize) = \convmapEstimate_0 \approx \convmap + \noise_0 ,
\end{equation}
and therefore,
\begin{equation}
    \pnpmassnoisy(\tilde\shearmapRand,\, \stepsizeRand) \approx \convmapRand_0,
\end{equation}
where $\convmapRand_0$ has been introduced in \eqref{eq:convmaprand0}.
Plugging this into \eqref{eq:approx_pmean_denoiser} and \eqref{eq:approx_pvar_denoiser}, replacing $\convmap_0$ by $\convmapEstimate_0$, we obtain
\begin{equation}
\deepmodelTrained(\convmapEstimate_0,\, \stepsize) \approx \condexpval{\convmapRand}{\pnpmassnoisy(\tilde\shearmapRand,\, \stepsizeRand) = \pnpmassnoisy(\tilde\shearmap,\, \stepsize),\, \stepsizeRand = \stepsize}
\end{equation}
and
\begin{equation}
\deepmodelbisTrained(\convmapEstimate_0,\, \stepsize) \approx \condvar{\convmapRand}{\pnpmassnoisy(\tilde\shearmapRand,\, \stepsizeRand) = \pnpmassnoisy(\tilde\shearmap,\, \stepsize),\, \stepsizeRand = \stepsize}.
\end{equation}
It is then sufficient that $\pnpmassnoisy(\cdot,\, \stepsize)$ is injective to deduce that
\begin{equation}
    \deepmodelTrained(\convmapEstimate_0,\, \stepsize) \approx \condexpval{\convmapRand}{\tilde\shearmapRand=\tilde\shearmap}
\label{eq:pmean_pnpmass}
\end{equation}
and
\begin{equation}
    \deepmodelbisTrained(\convmapEstimate_0,\, \stepsize) \approx \condvar{\convmapRand}{\tilde\shearmapRand=\tilde\shearmap}.
\label{eq:pvar_pnpmass}
\end{equation}

The overall algorithm implementing the PnPMass iteration is summarized in Fig.~\ref{fig:pnpmass} and Algorithm~\ref{algo:pnpmass}.%

\input{inputs/flowchart}

\begin{algorithm}[t]
    \caption{PnPMass with uncertainty quantification}
    \begin{algorithmic}[1]
    \Require Noise covariance matrix $\covmatrNoise$;
    \Require Power spectrum $\powerspectrum$ (residual version only, Sect.~\ref{subsec:pnp_respnp});
    \Require Step-size $\stepsize$ and number of iterations $\niter$;
    \Require Denoiser $\deepmodelTrained$ trained on white Gaussian noise;
    \Require Order-2 moment network $\deepmodelbisTrained$.
    \Statex
    \State\textbf{input:} Noisy shear map $\tilde\shearmap$.
    \If{residual version}
        \State Compute Gaussian component: $\convmapGauss \gets \iterativewiener{\powerspectrum}(\tilde\shearmap)$;
        \State Update input (residual): $\tilde\shearmap \gets \tilde\shearmap - \tilde\convToShear\convmapGauss$.
    \EndIf
    \Statex
    \State \textbf{initialize:} $\convmap^{(0)} = \bzero$.
    \For{$i = 1, \ldots, \niter$}
        \State Forward step: $\convmap_0^{(i)} \gets \forwardop{\tilde\shearmap}\bigl(
            \convmap^{(i-1)},\, \stepsize
        \bigr)$ ;
        \State Backward step: $\convmap^{(i)} \gets \deepmodelTrained\bigl(
            \convmap_0^{(i)},\, \stepsize
        \bigr)$.
    \EndFor
    \State Set point estimate: $\convmapEstimate \gets \convmap^{(\niter)}$.
    \Statex
    \State Forward step: $\convmapEstimate_0 \gets \forwardop{\tilde\shearmap}(
        \convmapEstimate,\, \stepsize
    )$;
    \State Backward step (variance estimate): $\varmapEstimate \gets \deepmodelbisTrained(
        \convmapEstimate_0,\, \stepsize
    )$.
    \Statex
    \If{residual version}
        \State Update point estimate: $\convmapEstimate \gets \convmapEstimate + \convmapGauss$.
    \EndIf
    \State \Return $(\convmapEstimate,\, \varmapEstimate)$.
    \end{algorithmic}
    \label{algo:pnpmass}
\end{algorithm}

\subsubsection{Uncalibrated pre-estimation of per-pixel error bars}
\label{subsubsec:uq_mn_errbars}

Our goal is to provide per-pixel error bars with marginal and non-asymptotic coverage guarantees. {More precisely}, given a desired miscoverage rate $\alpha \in \zeroexcloneexcl$, we wish to estimate $\convmapEstimateLow$ and $\convmapEstimateHigh$ such that for {every} $k \in \smalloneto{\imgsize^2}$,
\begin{equation}
    \condproba{\convmapRand[k] \notin \interval{\convmapEstimateLow[k]}{\convmapEstimateHigh[k]}}{\tilde\shearmapRand = \tilde\shearmap} \leq \alpha.
\label{eq:miscoveragerate}
\end{equation}
To this end, we started with a rough estimate based on a Gaussian assumption. This was then calibrated (see Sect.~\ref{subsec:uq_cqr}). More precisely, we started with the following quantiles $\convmapEstimateLow$ and $\convmapEstimateHigh$ from a Gaussian distribution of mean $\convmapEstimate$ and variance $\varmapEstimate$, which are provided by Algorithm~\ref{algo:pnpmass}:
\begin{equation}
    \convmapEstimateHighLow[k]
        := \convmapEstimate[k] \pm \cdfg^{-1}\left(
        1 - \frac{\alpha}{2}
    \right) \, \stdmapEstimate[k],
\label{eq:errbars}
\end{equation}
where $\cdfg$ is the cumulative distribution function (CDF) of the normal distribution $\MN(0,\, 1)$, and $\stdmapEstimate = \sqrt{\varmapEstimate}$ is the standard deviation estimate. If the conditional distribution of $\convmapRand[k]$, given $\tilde\shearmap$, were $\MN(\convmapEstimate[k],\, \varmapEstimate[k])$, then \eqref{eq:miscoveragerate} would be satisfied by construction. This is not the case, however, and hence the need of a calibration step that we describe below.

\subsection{Calibration with conformal predictions}
\label{subsec:uq_cqr}

In Sect.~\ref{subsec:uq_mn} the pre-estimates of the per-pixel error bars were based on several assumptions entailing several approximations, as listed below.
\begin{enumerate}[label=$\bullet$]
    \item \eqref{eq:approx_pmean_denoiser} and \eqref{eq:approx_pvar_denoiser} state that the networks $\deepmodelTrained$ and $\deepmodelbisTrained$ are only approximations of the posterior mean and variance given $\convmapRand_0 = \convmap_0$, respectively. As discussed there, the approximation quality depends on factors such as the complexity of the neural networks used and the quality and diversity of the training data. The method therefore primarily captures aleatoric uncertainty, while epistemic uncertainty, associated with model expressiveness or limited knowledge, remains unaddressed.
    \item \eqref{eq:fixedpointcondition_rand} is yet another approximation of the point estimate which, as we argued there, relies on the assumption that the fixed point $\convmapEstimate$ of PnPMass lies close enough to $\convmap$. In practice, this is not necessarily true, and thus bias, that is, nonzero residuals in \eqref{eq:fixedpointcondition_rand}, may still be present. This is especially true near masked regions where the S/N is very low.
    \item The per-pixel error bars have been derived under a Gaussian assumption for the posterior distribution, making them only rough uncertainty estimates.
\end{enumerate}

For all these reasons, the miscoverage inequality as stated in \eqref{eq:miscoveragerate} is not guaranteed to hold. In this section, we present a calibration method directly inspired by CQR \citep{Romano2019}, for which we can obtain (marginal) coverage guarantees that are valid for finite samples (non-asymptotic) and do not depend on any assumption on the models or the data distribution (distribution free). The choice of this method over other calibration approaches such as risk-controlling prediction sets \citep{Angelopoulos2022c}, based on concentration inequalities, is motivated by the fact that it explicitly avoids overconservative error bands, making it suitable in a context in which the tightest possible error bars are required. We refer to \citet{LetermeDistributionfreeUncertainty2025} for a detailed discussion of this topic.

In Sect.~\ref{subsubsec:uq_cqr_general} we review the general principles of CQR for mass mapping. Then, we propose in Sect.~\ref{subsubsec:uq_cqr_minsize} a method for minimizing the size of the error bars without the need for a model selection on a validation set.

\subsubsection{General principles}
\label{subsubsec:uq_cqr_general}

Let $\calibrationsetMMLifted$ denote a calibration set of size $\sizeCalibrationset \in \mathN$. For each pair $(\tilde\shearmap_i,\, \convmap_i)$, we denote by $(\convmapEstimateLow_i,\, \convmapEstimateHigh_i)$ the lower and upper bounds obtained using \eqref{eq:errbars}, from which a vector of conformity scores, denoted by $\calibparamVec_i \in \mathR^{\imgsize^2}$, can be computed as follows. For each pixel $k \in \smalloneto{\imgsize^2}$,
\begin{equation}
    \calibparamVec_i[k] = \max\left(
        \convmapEstimateLow_i[k] - \convmap_i[k],\, \convmap_i[k] - \convmapEstimateHigh_i[k]
    \right).
\label{eq:conformity_score}
\end{equation}
Then, we introduce the calibration vector $\calibparamVec$, obtained by taking, in each pixel $k$, the $(1 - \alpha)(1 + 1/\sizeCalibrationset)$-th empirical quantile of $(\calibparamVec_i[k])_{i=\sizeTrainingset + 1}^{\sizeTrainingset + \sizeCalibrationset}$. The term $(1 + 1/\sizeCalibrationset)$, which accounts for finite-sample correction, imposes a lower bound to the target error rate: $\alpha \geq 1 / (\sizeCalibrationset + 1)$.
The calibrated lower and upper bounds are finally defined as%
\begin{equation}
    \convmapEstimateLowCQR := \convmapEstimateLow - \calibparamVec \qqand \convmapEstimateHighCQR := \convmapEstimateHigh + \calibparamVec.
\label{eq:calibrated_bounds}
\end{equation}

We denote by $\convmapEstimateLowCQRRand$ and $\convmapEstimateHighCQRRand$ the random vectors from which $\convmapEstimateLowCQR$ and $\convmapEstimateHighCQR$ are drawn. Their random nature comes from that of $\tilde\shearmapRand$, from which the noisy observation is drawn, as well as from $\calibrationsetMMLiftedRand$, from which the calibration set is drawn. If $\calibrationsetMMLiftedRand \cup \{(\tilde\shearmapRand,\, \convmapRand)\}$ are exchangeable (\eg, independent and identically distributed), and the conformity scores are almost surely distinct, then we have the following (marginal) miscoverage guarantee at pixel $k$:
\begin{equation}
    \alpha - \frac1{\sizeCalibrationset + 1} \leq \Proba\left\{
        \convmapRand[k] \notin \interval{\convmapEstimateLowCQRRand[k]}{\convmapEstimateHighCQRRand[k]}
    \right\} \leq \alpha.
\label{eq:miscoveragerate_cqr}
\end{equation}
By averaging over pixels, we deduce that the expected proportion of pixels whose mass estimates fall outside the corresponding predicted error interval also verifies the bounds in \eqref{eq:miscoveragerate_cqr}.

The bound in \eqref{eq:miscoveragerate_cqr} presents a major difference compared to \eqref{eq:miscoveragerate}. Whereas the latter is conditional given $\tilde\shearmap$, the former is marginalized over all possible outcomes of $\tilde\shearmapRand$. Consequently, some inputs may more likely lead to errors than others. Conformal prediction with conditional guarantees is a very challenging problem and has only been solved in specific cases \citep{GibbsConformalPrediction2025}. We leave this for future work. In addition, the coverage guarantee provided in \eqref{eq:miscoveragerate_cqr} is also marginalized over the distribution of all calibration sets. Therefore, accidentally picking an out-of-distribution calibration set may lead to a higher error rate than expected. This effect can be mitigated by increasing the size of the calibration set.

The calibration procedure is not independent of the covariance matrix $\covmatrNoise$, unlike the training phase. However, we point out that the former is cheaper by several orders of magnitude than the latter, computationally speaking, as detailed in Sect.~\ref{subsec:results_times}.%

\subsubsection{Minimizing the size of the error bars}
\label{subsubsec:uq_cqr_minsize}

We can easily show that the miscoverage guarantee from \eqref{eq:miscoveragerate_cqr} still holds if the pre-calibration bounds $\convmapEstimateHighLow$, introduced in \eqref{eq:errbars}, are replaced by
\begin{equation}
    \convmapEstimateHighLow_{\multfactbounds}[k]
        := \convmapEstimate[k] \pm \multfactbounds \, \cdfg^{-1}\left(
        1 - \frac{\alpha}{2}
    \right) \, \stdmapEstimate[k]
\label{eq:errbars_0}
\end{equation}
for any multiplicative constant $\multfactbounds > 0$. The calibration vector is then recomputed for each $\multfactbounds$, and is denoted by $\calibparamVec_{\multfactbounds}$. The mean error bar size can then be minimized over $\multfactbounds$:%
\begin{equation}
    \multfactbounds^\ast \in \argmin_{\multfactbounds > 0} \multfactbounds \errbarsize + \calibparam_{\multfactbounds},
\end{equation}
where $\errbarsize > 0$ denotes the mean half-size of the precalibration error bars (averaged over all test images and all pixels), and $\calibparam_{\multfactbounds} \in \mathR$ denotes the mean value of the calibration vector $\calibparamVec_{\multfactbounds}$, which is a non-increasing function of $\multfactbounds$. 
This procedure is computationally inexpensive because $\convmapEstimate_i$ and $\stdmapEstimate_i$ are computed only once. For each candidate $\multfactbounds$, only the conformity scores $\calibparamVec_{\multfactbounds,\, i}$ and their corresponding quantiles must be recomputed in order to obtain $\calibparamVec_{\multfactbounds}$.

\section{Experimental settings}
\label{sec:experiments}

\subsection{Datasets}
\label{subsec:experiments_datasets}

In order to generate the training, validation, calibration, and test sets, the ground-truth convergence maps $\convmap_i$ were obtained using the $\kappa$TNG dataset of mock convergence maps, as described below. The mass maps were generated from the full redshift distribution, following the same procedure as \citet{LetermeDistributionfreeUncertainty2025}.%

For the training and validation sets, the noisy inputs $\convmap_{0,\, i}$ were obtained from $\convmap_i$ by adding a white Gaussian noise with a standard deviation uniformly distributed between $0$ and $0.2$. The choice for this particular interval is explained in Sect.~\ref{subsubsec:experiments_impl_pnpmass}. At each training epoch, a new noise realization was drawn, with a newly selected noise level.
For the calibration and test sets, the noisy shear maps $\shearmap_i$ were computed from $\convmap_i$ using~\eqref{eq:invprob}. The covariance matrix $\covmatrNoise$ was obtained by binning the weak-lensing shear catalog from \citet{Schrabback2010} (bright catalog only) at the $\kappa$TNG resolution ($0.29$ arcmin per pixel), and then applying~\eqref{eq:covmatrnoise}. The intrinsic standard deviation $\stdval_{\ellip}$ was estimated from the measured ellipticities and set to $0.39$.

This catalog is based on the Cosmic Evolution Survey (COSMOS, \citealp{Scoville2007}) obtained with the Hubble Space Telescope, as well as photometric redshift measurements from \citet{Mobasher2007}. It contains an average of $32$ galaxies per square arcminute over a wide range of redshifts, which is consistent with what is expected from forthcoming surveys such as Euclid. This number excludes the galaxies with redshifts above $2.6$ for the sake of consistency with the $\kappa$TNG dataset. It also ignores the masked pixels, without any observed galaxy.
An example of a mock convergence map $\convmap_i$ from $\kappa$TNG, and the corresponding noisy shear map $\shearmap_i$, is provided in Appendix~\ref{appendix:visual_representations}, Fig.~\ref{fig:convshear}.%

The $\kappa$TNG cosmological hydrodynamic simulations \citep{Osato2021} include realizations of $5 \times 5 \deg^2$ convergence maps for various source redshifts up to $z = 2.6$, at a $0.29$ arcmin per pixel resolution, assuming a flat $\Lambda$CDM universe. Only one set of cosmological parameters was used to run the simulations: $H_0 = 67.74 \km\cdot\s^{-1}\cdot\Mpc^{-1}$, $\Omega_{\baryon} = 0.0486$, $\Omega_{\matter} = 0.3089$, $n_{\spectralindex} = 0.9667$, and $\sigma_8 = 0.8159$. For training and validation, we used $98$ and $2$ nearly independent realizations generated from a single lensing potential, respectively, that we randomly rotated and cropped to $384 \times 384$ pixels until we reached $70\,560$ images for training and $1\,440$ for validation. For calibration and testing, we used $57$ and $43$ nearly independent realizations generated from another lensing potential, respectively. The test set was then obtained by extracting $512$ non-overlapping\footnote{Accounting for the COSMOS boundaries.} crops of size $384 \times 384$ pixels, whereas the calibration set was obtained by following the augmentation procedure described above, with a total of $1\,024$ images.

\subsection{Other implementation details}

\subsubsection{Parameters for PnPMass}
\label{subsubsec:experiments_impl_pnpmass}

With the specific value of $\covmatrNoise$ introduced in Sect.~\ref{subsec:experiments_datasets}, the step size $\stepsize$ must be chosen in the interval $\intervalexcllr{0}{0.176}$ to comply with~\eqref{eq:stepsize2}. To derive these bounds, $\lambda_{\max}$ and $\lambda_{\min}$ were computed using the power iteration method with $100$ iterations (see Appendix~\ref{appendix:powitlambdamin}). To assess the sensitivity of PnPMass to the step size, we tested several values of $\stepsize$, set to $50\%$, $75\%$, and $100\%$ of $\stepsizeMax := (1+\rho)/\lambda_{\max}$, where $\rho$ was taken arbitrarily close to $1$.

According to \eqref{eq:covmatrwhitenoise}, the denoiser must be able to operate at a noise level equal to the step size $\stepsize$. For flexibility reasons, we chose to train the model on a realistic range of noise levels including the above interval, following the method described in Sect.~\ref{subsubsec:pnp_denoiserprops_noiseaware}. As detailed in Sect.~\ref{subsec:experiments_datasets}, we used a uniform distribution between $0$ and $0.2$.

For the PnPMass variant on non-Gaussian residuals described in Sect.~\ref{subsec:pnp_respnp}, the power spectrum $\powerspectrum$ was estimated on $2\,048$ ground-truth convergence maps from the training set. The Gaussian component $\convmapGauss$ was then computed using an iterative Wiener filtering method with $12$ iterations and a step-size of $1.61 \times 10^{-2}$.

\subsubsection{Uncertainty quantification parameters}
\label{subsubsec:experiments_impl_uq}

We selected a target error rate $\alpha$ such that $\cdfg^{-1}\!\left(
    1 - \alpha/2
\right) = 2$, see \eqref{eq:errbars}. This setting, referred to as $2\sigma$ confidence, corresponds to $\alpha \approx 4.55\%$. We note that higher confidence levels require larger calibration sets due to the finite-sample correction: the CQR algorithm selects the $(1 - \alpha)(1 + 1/\sizeCalibrationset)$-th empirical quantile of the conformity scores, which by definition must remain below $100\%$. Consequently, a $2\sigma$ confidence requires at least $\sizeCalibrationset = 21$ calibration samples, versus $\sizeCalibrationset = 370$ at $3\sigma$ and $\sizeCalibrationset = 15\,787$ at $4\sigma$. In our experiments, we chose $1\,024$ calibration examples, which is far above the minimum requirement of $21$ samples. According to \eqref{eq:miscoveragerate_cqr}, the probability of miscoverage after calibration with this setting ranges between $4.45\%$ and $4.55\%$.

\subsubsection{Denoiser training and PnPMass implementation}
\label{subsubsec:experiments_impl_training}

Our models are based on SUNet \citep{Fan2022}, a Swin transformer-based architecture with $7.2$ millions of trainable parameters, which we adapted to incorporate noise-level awareness. Following an approach similar to that used by \citep{ZhangPlugandPlayImage2022} for DRUNet, we introduced an additional input channel that encodes the noise standard deviation. Since our models were trained with stationary noise, this channel was filled with the constant value $\stdval_i$ for each input $\convmap_{0,\, i}$.

We trained four models in total: two order-1 models $\deepmodelTrained$ for the point estimate, and two order-2 models $\deepmodelbisTrained$ for the variance (see Algorithm~\ref{algo:pnpmass}). In each case, we trained one model for standard PnPMass and another one for the variant on non-Gaussian residuals (see Sect.~\ref{subsec:pnp_respnp}), using Adam as the optimizer and the squared $\ell^2$ loss. Training was run for $100$ epochs, with a batch size of $16$. The learning rate was initialized at $10^{-3}$ and reduced by a factor of $10$ four times during training.
For the PnPMass variant on non-Gaussian residuals, the Gaussian components $\convmapGauss_i$ were computed from $\convmap_{0,\, i}$ using one-step Wiener filtering, which is feasible because the noise is stationary, with the power spectrum $\powerspectrum$ introduced in Sect.~\ref{subsubsec:experiments_impl_pnpmass}.

In $\deepmodelbisTrained$, the final mean-centering layer was replaced with a rectified linear unit (ReLU) activation layer to enforce nonnegative outputs. To avoid vanishing gradients during training, however, this layer was only applied at inference time. The training set was built from the same dataset as we used for training $\deepmodelTrained$, with the ground truth derived from \eqref{eq:groundtruth_order2}.
Our PnPMass implementation is based on the DeepInverse library \citep{TachellaDeepInverse2025},\footnote{\url{https://deepinv.github.io/}} which is built on PyTorch.

\section{Results and discussion}
\label{sec:results}

\begin{figure}
	\centering
    \includegraphics[width=\columnwidth]{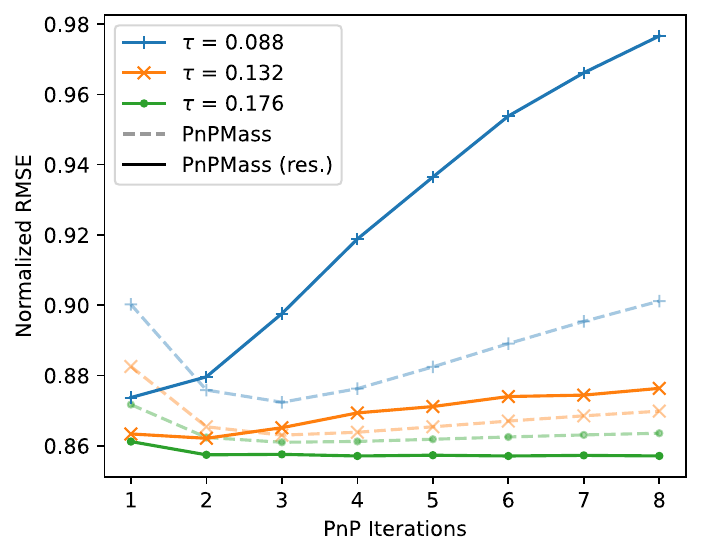}
    \vspace{-15pt}
    \caption{Evolution of the mean RMSE along PnPMass iterations, computed on the $\kappa$TNG test set ($512$ images) for various step sizes $\stepsize$. The results have been normalized \wrt to the $\ell^2$-norm of each image.}
\label{fig:pnpmass_iterations}
\end{figure}

\subsection{Reconstruction accuracy}

\subsubsection{PnPMass iterations}

We ran PnPMass for eight iterations and plotted the evolution of the normalized root-mean-square error (RMSE) over the test set, restricted to active pixels, for various step sizes $\stepsize$ for the standard PnPMass and its variant on non-Gaussian residuals. The results are shown in Fig.~\ref{fig:pnpmass_iterations}.
Throughout the discussion, we refer to this metric as the reconstruction accuracy.

The RMSE decreases over the first iterations before either stabilizing or increasing again, depending on the step size. Moreover, the reconstruction accuracy of PnPMass and its variant on non-Gaussian residuals increases as $\stepsize$ increases in the devised range. PnPMass on non-Gaussian residuals is notably more sensitive to the choice of $\stepsize$ than standard PnPMass. However, since the former incorporates prior knowledge of the Gaussian structure of the background mass distribution, it achieves the best reconstruction accuracy compared to the standard PnPMass at $\stepsize = \stepsizeMax$. Based on bootstrapping, the expected improvement after eight iterations lies between $6.04 \times 10^{-3}$ and $6.88 \times 10^{-3}$ with $95\%$ confidence.

From now on, we focus on the results obtained after eight iterations, at $\stepsize = \stepsizeMax$. This choice was motivated by empirical observations because the reconstruction accuracy did not improve further beyond a few iterations.

\subsubsection{Benchmark against other approaches}

We compared the performance of the Wiener, MCALens, DeepMass, and PnPMass algorithms in terms of reconstruction accuracy (RMSE on active pixels). The results are displayed in Table~\ref{table:rmse}.

\begin{table}
    \caption{Reconstruction accuracy.}             
    \label{table:rmse}
    \centering
    \setlength{\tabcolsep}{4pt}
    \vspace{-5pt}
    \small
    \begin{centering}
    \renewcommand{\arraystretch}{1.1}
    \begin{tabular}{ r | c }
        \hline\hline
        \multicolumn{1}{c|}{Method} & \multicolumn{1}{c}{Normalized RMSE ($\times 10^{-1}$)} \\
        \hline
        Wiener$^\ast$ & $8.86 \pm 0.34$ \\
        MCALens$^{\dagger}$ & $8.74 \pm 0.44$ \\
        DeepMass$^\ddagger$ & $8.53 \pm 0.48$ \\
        \hline
        PnPMass  & $8.64 \pm 0.50$ \\
        PnPMass (res.) & $8.58 \pm 0.49$ \\
        \hline
    \end{tabular}
    \end{centering}
    \begin{minipage}{\linewidth}
        \vspace{5pt}
        \footnotesize
        \textbf{Notes.} 
        For each test example, the RMSE was computed over active pixels and normalized with respect to the $\ell^2$-norm of the ground-truth image. Then, the mean and standard deviation were evaluated across the test set ($512$ images).
        $^\ast$Wiener was run for $12$ iterations with a power spectrum estimated from $2\,048$ ground-truth $\kappa$-maps in the training set.
        $^\dagger$MCALens was run for $16$ iterations (early stopping), with a $4.5\sigma$ detection threshold and a step size of $0.088$ for the non-Gaussian component.
        $^\ddagger$DeepMass was trained for $20$ epochs on a UNet architecture with $280$k parameters, following \citet{Jeffrey2020}.
    \end{minipage}
\end{table}

We found that our methods outperform the classical approaches (Wiener and MCALens), but fall slightly behind the reconstruction accuracy achieved by DeepMass. A likely explanation is that DeepMass has been fine-tuned for a specific mask $\mask$ and noise covariance $\covmatrNoise$, whereas PnPMass is applicable across a wide range of configurations. However, as discussed in Sect.~\ref{subsubsec:background_relatedwork_dl}, this level of performance comes at the cost of retraining for each new observation, which limits its practicality. We therefore argue that PnPMass offers a very good trade-off between flexibility and reconstruction accuracy.

\subsection{Uncertainty quantification results}

For each mass-mapping method introduced above and each example in the $\kappa$TNG test set, we computed the lower and upper bounds before and after calibration. Then, we measured the empirical miscoverage rate, that is, the percentage of active pixels falling outside the predicted bounds.
For Wiener and MCALens, the error bars were simply initialized to $0$ before calibration. Regarding DeepMass and PnPMass, the error bars before calibration were obtained using appropriate order-2 moment networks, as described in Sects.~\ref{subsubsec:uq_mn_deepmass} and \ref{subsubsec:uq_mn_pnp}, respectively. We then applied calibration while minimizing the size of the error bars, as described in Sect.~\ref{subsubsec:uq_cqr_minsize}.
The corresponding results, computed on the test set, are depicted in Figs.~\ref{fig:summary} and \ref{fig:miscoveragerate} (Appendix~\ref{appendix:plots}).

\begin{figure}
    \centering
    \includegraphics[trim=0 0 0 7pt, clip, width=\columnwidth]{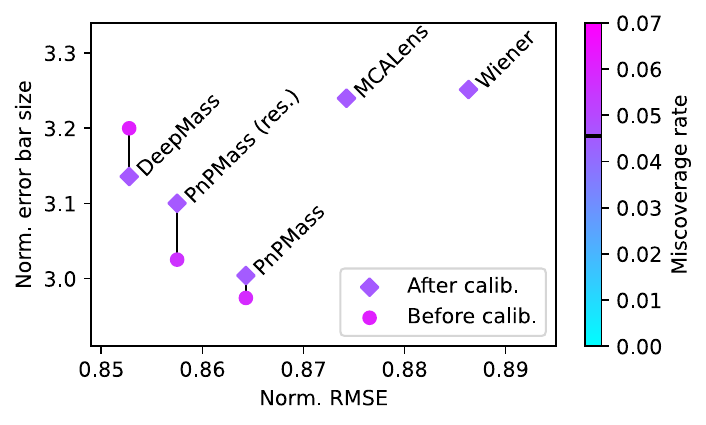}
    \vspace{-15pt}
    \caption{Mean length of prediction intervals plotted against reconstruction accuracy before and after calibration. Both metrics have been normalized \wrt to the $\ell^2$-norm of each image. The marker colors indicate empirical miscoverage rates, the target $\alpha$ being marked on the color bar. The metrics are computed over the $\kappa$TNG test set ($512$ examples). All mass-mapping algorithms were run with the same parameters as in Table~\ref{table:rmse}. For Wiener and MCALens, only the post-calibration miscoverage rate is displayed, since their error bars were initialized to $0$.}
    \label{fig:summary}
\end{figure}

\begin{figure*}
    \centering
    \small
    \vspace{-8pt}
    \begin{tikzpicture}

        \node (mmgan) {\includegraphics[height=0.22\textwidth, trim={0 0 50pt 0}, clip]{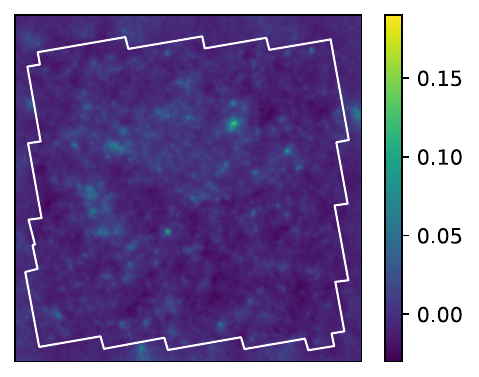}};
        \node[anchor=west, xshift=-5pt] (deepmass) at (mmgan.east) {\includegraphics[height=0.22\textwidth, trim={0 0 50pt 0}, clip]{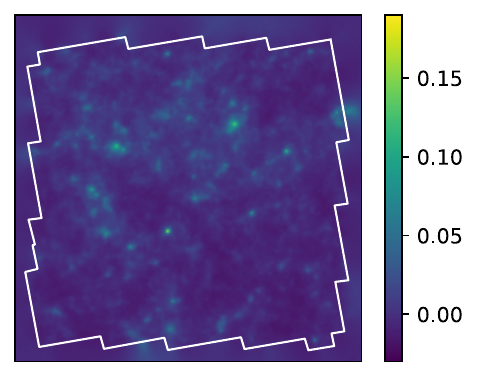}};
        \node[anchor=west, xshift=-5pt] (pnpmass) at (deepmass.east) {\includegraphics[height=0.22\textwidth, trim={0 0 50pt 0}, clip]{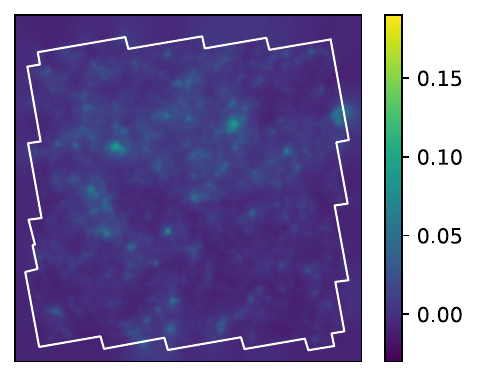}};
        \node[anchor=west, xshift=-5pt] (respnpmass) at (pnpmass.east) {\includegraphics[height=0.22\textwidth]{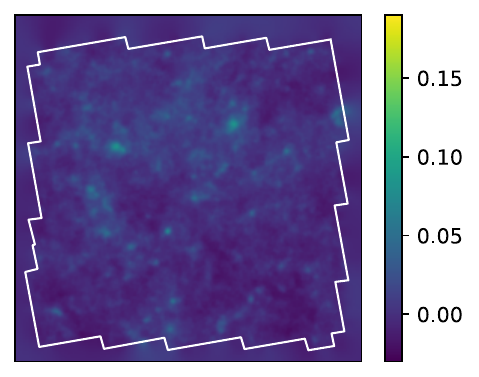}};

        \node (mmgan caption) at (mmgan.north) {MMGAN$^\ast$};
        \node (deepmass caption) at (deepmass.north) {DeepMass$^{\ast\ast}$};
        \node (pnpmass caption) at (pnpmass.north) {\textbf{PnPMass}};
        \node[xshift=-12pt] (respnpmass caption) at (respnpmass.north) {\textbf{PnPMass (res.)}};

        \end{tikzpicture}
    \vspace{-8pt}
    \caption{Reconstruction of the COSMOS field with several deep-learning-based mass-mapping methods. For consistency with MMGAN (left), we concatenated the bright and faint catalogs before binning to the $\kappa$TNG resolution. $^\ast$Directly provided by \citet{WhitneyGenerativeModellingMassmapping2025}. $^{\ast\ast}$DeepMass was retrained to take the noise level and mask from the bright + faint catalogs into account.}
    \label{fig:visual_representations_cosmos}
\end{figure*}

\subsubsection{Accuracy versus error bar size}
\label{subsubsec:results_uq_errbars}

After calibration, all methods achieved the same miscoverage rate, as guaranteed by \eqref{eq:miscoveragerate_cqr}. However, the mean error bar size differed in the methods, as discussed below.

First, let us focus on the comparison of our approach with Wiener and MCALens. PnPMass in the standard and residual versions produces smaller error bars than the two model-driven methods (see the y-axis in Fig.~\ref{fig:summary}).
We hypothesize that the combination of higher reconstruction accuracy and adaptive pre-calibrated error bars obtained with order-2 moment networks helps to distinguish between smooth regions and peak areas. This might in turn result in a more refined uncertainty quantification by correctly assigning higher uncertainties near peaks.
In contrast, Wiener and MCALens produce error bars that are independent of the input images: they are initialized at $\bzero$ and subsequently calibrated using a parameter that only depends on the pixel location.

More surprisingly, PnPMass also produces smaller error bars than DeepMass, but they are slightly less accurate. This result is statistically significant because it is observed for all $512$ test images. To explain this phenomenon, we propose the following conjecture: DeepMass detects more peaks than PnPMass, which improves reconstruction accuracy, but is also more prone to hallucinations (\ie, false detections), which negatively affects the confidence level. In other words, the optimal precision-recall balance for peak detection depends on the metric that is to be prioritized. Consequently, a trade-off arises between reconstruction accuracy and error bar size. The same reasoning applies when we compare the standard and residual versions of PnPMass.

It is worth noting that error bars might in principle have been initialized at $\bzero$ for all mass-mapping methods, avoiding the need to train an order-2 moment network for DeepMass and PnPMass. However, this would ultimately result in larger calibrated error bars to compensate for the poorer initialization. A plot similar to Fig.~\ref{fig:summary}, with zero-initialization for all methods, is shown in Appendix~\ref{appendix:plots}, Fig.~\ref{fig:summary_zeroinit}.

\subsubsection{Effect of the calibration on the confidence intervals}

For all methods based on moment networks, the pre-calibration error bars yield a miscoverage rate slightly above the target $\alpha$. Interestingly, calibration reduced the error bar size for DeepMass, which may seem counterintuitive. This effect arises from the optimization strategy described in Sect.~\ref{subsubsec:uq_cqr_minsize}. Our interpretation is that while error bars were underestimated for most pixels, some regions, particularly around peaks, were overly conservative, leading to an overall overestimation that calibration was able to correct.

\subsection{Mass-mapping visualization}

We applied the two versions of PnPMass to the COSMOS shear field, which we binned at the $\kappa$TNG resolution. A visual representation of the two reconstructions is shown in Fig.~\ref{fig:visual_representations_cosmos}, together with DeepMass and MMGAN reconstructions.\footnote{We did not reproduce MMGAN ourselves; the reconstruction was downloaded from the Zenodo repository provided by the authors.} Additional visuals are provided in Appendix~\ref{appendix:visual_representations}, Fig.~\ref{fig:visual_representations_cosmos_bis}, including the position of known x-ray clusters, as well as uncertainty maps (standard deviation per pixel).

A visual example of a reconstructed convergence map from a $\kappa$TNG simulation, along with its predicted standard deviation and bounds, is shown in Appendix~\ref{appendix:visual_representations}, Fig.~\ref{fig:visual_representations}.
While the post-calibration miscoverage rate remains constant on average, there are situations in which a given method may systematically fail to estimate uncertainty correctly. This is particularly noticeable for the Wiener solution around peak values: the blue spots in the bottom image indicate miscoverage. In contrast, our methods better account for peak values because the order-2 moment networks tends to predict a higher uncertainty in these regions. However, conformal prediction does not guarantee control of the miscoverage rate in specific situations such as peak values. One possible improvement would be to employ conditional conformal prediction \citep{GibbsConformalPrediction2025}, which we leave for future work.

\subsection{Computation times}
\label{subsec:results_times}

\begin{table}
    \caption{Computation times for $\kappa$TNG.}             
    \label{table:time}
    \centering
    \setlength{\tabcolsep}{4pt}
    \vspace{-5pt}
    \small
    \begin{centering}
    \renewcommand{\arraystretch}{1.1}
    \begin{tabular}{ r | r r | r r }
        \hline\hline
        & \multicolumn{2}{c|}{Training}  & & \\
        \multicolumn{1}{c|}{Method} & \multicolumn{1}{c}{Order~1} & \multicolumn{1}{c|}{Order~2} & \multicolumn{1}{c}{Calibration} & \multicolumn{1}{c}{Inference} \\
        \hline
        Wiener & -- & -- & $9''$ & $6''$ \\
        MCALens & -- & -- & $44''$ & $24''$ \\
        DeepMass$^\ddagger$ & $07$:$27'$ & $46$:$33'$ & $18''$ & $7''$ \\
        \hline
        PnPMass & $31$:$41'$ & $47$:$20'$ & $1'42''$ & $50''$ \\
        PnPMass (res.) & $33$:$43'$ & $43$:$44'$ & $1'48''$ & $54''$ \\
        \hline
    \end{tabular}
    \end{centering}
    \begin{minipage}{\linewidth}
        \vspace{5pt}
        \footnotesize
        \textbf{Notes.}
        Training times were estimated for the first- and second-order moment networks on a single NVIDIA Quadro RTX 6000 GPU. Calibration and inference were performed on a single NVIDIA GeForce GTX 1080 Ti GPU, using $1\,024$ and $512$ examples, respectively. Inference includes the computation of a point estimate, the initialization of error bars, and their calibration using the parameters computed at the previous stage.
        The training times are expressed in hours, not minutes.
        $^\ddagger$The order-2 moment network for DeepMass was trained for $100$ epochs, as the validation curve indicated that training was not yet complete after $20$ epochs, in contrast to the order-1 network.
    \end{minipage}
\end{table}

The computation times for training, calibration, and inference are reported in Table~\ref{table:time}.
While calibration and inference with PnPMass are slower than with DeepMass due to its iterative nature, these differences are negligible compared to training times, which are longer by several orders of magnitude than the inference times. This highlights one of the main advantages of our approach: unlike DeepMass and MMGAN, training is required only once.

\section{Conclusion}
\label{sec:concl}

We introduced PnPMass, a PnP-based mass-mapping method that combines the strengths of DeepMass and DeepPosterior, namely, flexibility (no need to retrain for each new observation) and fast inference, while maintaining competitive reconstruction accuracy. We further proposed a fast UQ strategy for PnPMass based on moment networks, followed by a conformal prediction to mitigate several sources of bias, including model uncertainty. Our results show that PnPMass yields the smallest calibrated error bars of all evaluated mass-mapping methods. Together, these contributions make our approach well suited to handle the large data volumes expected from upcoming stage-IV surveys such as Euclid and Rubin.

Several directions remain open to fully integrate our framework into weak-lensing pipelines.
First, further investigations are needed to assess the robustness of our approach across different cosmologies. Our experiments focused on $\kappa$TNG, a dataset simulated from a single set of cosmological parameters. A more realistic setting would involve datasets spanning a range of cosmological models.
Second, although conformal prediction provides coverage guarantees on average, it may fail to control miscoverage around peak structures, which are particularly important for cosmological parameter inference \citep{TersenovImpactWeaklensingMassmapping2025}. A promising avenue to address this limitation is to employ conditional conformal prediction \citep{GibbsConformalPrediction2025}, rather than the standard version used here.
Finally, it will be crucial to assess how uncertainties estimated at the mass-mapping stage propagate to cosmological parameter inference and affect the resulting FoM. In particular, this analysis could clarify whether the residual variant should be preferred over the standard version of PnPMass, as optimal cosmological parameter estimation may depend on a trade-off between reconstruction accuracy and size of the confidence intervals.

\section*{Data availability}

To ensure reproducibility, all software, scripts, and notebooks used in this study are available on GitHub, at \url{https://github.com/hubert-leterme/weaklensing_uq.git}.

\begin{acknowledgements}
    This work was funded by the TITAN ERA Chair project (contract no.\@ 101086741) within the Horizon Europe Framework Program of the European Commission, and the French Agence Nationale de la Recherche (ANR-22-CE31-0014-01 TOSCA and ANR-18-CE31-0009 SPHERES).
\end{acknowledgements}

%
\bibliographystyle{bibtex/aa} 
\bibliography{bibtex/refs} 
%

\begin{appendix}

\section{Proof of Proposition~1}
\label{appendix:proof_fixedpoint}

We start with consider this preparatory lemma.
\begin{lemma}
    \label{lemma:contraction}
    Suppose that \ref{assum:3} and \ref{assum:3'} hold. Then, for $\rho \in \intervalexclr{\rho^\star}{1}$,
    the forward operator $\forwardop{\tilde\shearmap}(\cdot,\, \stepsize)$ is $\rho$-contractive on $\orthcomplBA$ if, and only if, the step-size $\stepsize$ satisfies \eqref{eq:stepsize}.
    The smallest contraction factor $\rho = \rho^\star$ is obtained for $\tau = \tau^\star$, with%
    \begin{equation}
        \tau^\star := \frac{2}{\lambda_{\max}+\lambda_{\min}}.
    \label{eq:stepsize_lowerbound_rho}
    \end{equation}
\end{lemma}
\begin{proof}
    Denote for short the subspace $\mathV:=\orthcomplBA$. By construction, and since \ref{assum:3'} holds, we have:
    \begin{equation}
        \forwardop{\tilde\shearmap}(\cdot,\, \stepsize): \mathV \to \mathV.
    \label{eq:forwarop_subspace}
    \end{equation}
    Let $\convmap'$ and $\convmap'' \in \mathV$. The forward operator is an affine map, and we thus have
    \begin{align}
        \forwardop{\tilde\shearmap}(\convmap',\, \stepsize) - \forwardop{\tilde\shearmap}(\convmap'',\, \stepsize)
            = (\BBI - \stepsize\tilde\genericop\tilde\convToShear) \, (\convmap' - \convmap'') .
    \end{align}
    Since $\tilde\genericop\tilde\convToShear$ is symmetric positive definite on the orthogonal to its kernel, the Lipschitz constant $L_\tau$ of $\forwardop{\tilde\shearmap}(\cdot,\, \stepsize)$ on $\mathV$ is exactly given by 
    \begin{equation}
    L_\tau = \max(|1-\tau\lambda_{\min}|,|1-\tau\lambda_{\max}|) .
    \end{equation}
    Consequently, $\forwardop{\tilde\shearmap}(\cdot,\, \stepsize)$ is $\rho$-contractive on $\mathV$ if, and only if, $L_\tau \leq \rho$. Expanding this inequality, and considering that $\lambda_{\min} \leq \lambda_{\max}$, this is equivalent to \eqref{eq:stepsize}. The lower bound $\rho^\star$ introduced in \eqref{eq:lowerbound_rho} ensures the non-emptiness of the interval. Moreover, setting $\rho = \rho^\star$ reduces the range of possible step-size values to $\stepsize = \stepsize^\star$ such as introduced in \eqref{eq:stepsize_lowerbound_rho}. Finally, having $\rho < 1$ provides the contractive property of $\forwardop{\tilde\shearmap}(\cdot,\, \stepsize)$.
\end{proof}

Let us turn to the proof of Proposition~1.

\begin{proof}
According to \ref{assum:1} and \eqref{eq:forwarop_subspace}, both $\deepmodelTrained$ and $\forwardop{\tilde\shearmap}(\cdot,\, \stepsize)$ map $\mathV$ to $\mathV$.
Combining this with the non-expansiveness of $\deepmodelTrained$ on $\mathV$ \ref{assum:2}, we have, for all $\convmap'$ and $\convmap'' \in \mathV$,
\begin{align}
        \norm{\fixedpointiter{\tilde\shearmap}(\convmap',\, \stepsize) - \fixedpointiter{\tilde\shearmap}(\convmap'',\, \stepsize)}
            &\leq \norm{\forwardop{\tilde\shearmap}(\convmap',\, \stepsize) - \forwardop{\tilde\shearmap}(\convmap'',\, \stepsize)} .
\end{align}
Invoking now Lemma~1 and using the step-size choice in \ref{assum:4}, we get
\begin{align}
        \norm{\fixedpointiter{\tilde\shearmap}(\convmap',\, \stepsize) - \fixedpointiter{\tilde\shearmap}(\convmap'',\, \stepsize)}
            &\leq \rho \norm{\convmap' - \convmap''}.
\end{align}
Clearly, $\fixedpointiter{\tilde\shearmap}(\cdot,\, \stepsize)$ is a $\rho$-contraction on $\mathV$. Since $\mathV$ is obviously a complete metric space, we can now invoke the Banach-Picard fixed point theorem on $\fixedpointiter{\tilde\shearmap}(\cdot,\, \stepsize)$ to get that it has a unique fixed point $\convmapEstimate \in \mathV$. Moreover, using \ref{assum:1} and an induction argument, $(\convmap^{(k)})_{k \in \mathN} \subset \mathV$, and thus, the sequence converges linearly to $\convmapEstimate$ at the rate $\rho$.
\end{proof}

\section{Insights into convergence properties}
\label{appendix:convprops}

In this section, we sketch a mathematical framework for studying the convergence properties of our algorithm, which leads to constrain the noise level at which the denoiser must be trained, as explained in Sect.~\ref{subsec:pnp_smallresidual}.

Inserting \eqref{eq:invprob_lifted} into the fixed-point equation \eqref{eq:fixedpoint}, we infer that%
\begin{equation}\label{eq:convmapresidual}
\convmapEstimate =  \deepmodelTrained\Bigl(\convmap + \residualEstimate + \stepsize\tilde\genericop(-\tilde\convToShear\residualEstimate + \tilde\noise)\Bigr)  = \deepmodelTrained\Bigl(\convmap + (\BBI-\stepsize\tilde\genericop\tilde\convToShear)\residualEstimate + \noise_0\Bigr) ,
\end{equation}
where we denoted the residual $\residualEstimate := \convmapEstimate - \convmap$ and $\noise_0 = \stepsize\tilde\genericop\tilde\noise$. Observe that $\noise_0$ is a sample of $\noiseRand_0$ as defined in \eqref{eq:pnpnoise}.

Clearly, $\deepmodelTrained$ acts as a denoiser on the true convergence map $\convmap$ corrupted by the ``noise'' $(\BBI-\stepsize\tilde\genericop\tilde\convToShear)\residualEstimate + \noise_0$. Intuitively, one would like $(\BBI-\stepsize\tilde\genericop\tilde\convToShear)\residualEstimate$ to be small enough so that the dominating part of the noise would be that from $\noise_0$. This intuition is exact when $\convToShear=\BBI$ (\ie denoising), in which case, with the natural choice $\tilde\genericop=\BBI$ and $\tau=1$, \eqref{eq:convmapresidual} reads $\convmapEstimate = \deepmodelTrained\bigl(\convmap + \noise_0\bigr)$ with $\noise_0=\tilde\noise$. One should then train the denoiser with noise $\noise_0$.

To get some insight into this intuition in the general case, let us now bound the residual $\residualEstimate$. Let $\bzeta$ be a sample of the zero-mean random noise $\zetaRand$ used during training, and assume that the support of the distribution of $\convmapRand+\zetaRand$ is contained in $\orthcomplBA$.
Subtracting $\convmap$ to both sides of \eqref{eq:convmapresidual}, we have:
\begin{align}
\norm{\residualEstimate}
    &\leq \norm{\deepmodelTrained\Bigl(\convmap + (\BBI-\stepsize\tilde\genericop\tilde\convToShear)\residualEstimate + \noise_0\Bigr) \!-\! \deepmodelTrained\Bigl(\convmap+\bzeta\Bigr)} + \norm{\deepmodelTrained\Bigl(\convmap+\bzeta\Bigr) \!-\! \convmap} \nonumber\\
    &\leq \norm{(\BBI-\stepsize\tilde\genericop\tilde\convToShear)\residualEstimate + \noise_0 - \bzeta} + \norm{\deepmodelTrained\Bigl(\convmap+\bzeta\Bigr) - \convmap}
\label{eq:normresidual_2} \\
    &\leq \rho\norm{\residualEstimate} + \norm{\noise_0 - \bzeta} + \norm{\deepmodelTrained\Bigl(\convmap+\bzeta\Bigr) - \convmap} ,
\label{eq:normresidual_3} 
\end{align}
where we used \ref{assum:2} in \eqref{eq:normresidual_2}, and the triangle inequality and Lemma~1 (Appendix~\ref{appendix:proof_fixedpoint}) in \eqref{eq:normresidual_3}.
We then deduce, after taking the expectation with respect to the joint distribution of $(\convmapRand,\noiseRand,\zetaRand)$, that
\begin{align}
\Expval\left[\norm{\residualEstimateRand}\right] 
    &\leq (1-\rho)^{-1} \left(\,
        \Expval\left[
            \bignorm{\noiseRand_0 - \zetaRand}
        \right] + \Expval\left[
            \norm{\deepmodelTrained\bigl(
                \convmapRand+\zetaRand
            \bigr) - \convmapRand}
        \right]
    \,\right) \nonumber\\
    &\leq (1-\rho)^{-1} \left(
        \Expval\left[
            \bignorm{\noiseRand_0 - \zetaRand}^2
        \right]^{\frac12} \!+\! \Expval\left[
            \norm{\deepmodelTrained\bigl(
                \convmapRand+\zetaRand
            \bigr) - \convmapRand}^2
        \right]^{\frac12}
    \right) ,\label{eq:boundexpresidual}
\end{align}
where we used Jensen's inequality.

We aim to minimize the expected norm of the residual such as introduced above. Ideally, the denoiser $\deepmodelParameterized$ would be trained to minimize the population quadratic risk, which is nothing but the second term in the upper bound \eqref{eq:boundexpresidual}. In practice, however, training is performed using the empirical quadratic risk, which introduces an additional generalization error term that must be accounted for.
Regarding the first term in \eqref{eq:boundexpresidual}, it captures the discrepancy between the noise in training and that in the PnP iteration. We have
\[
\Expval\left[\norm{\noiseRand_0 - \zetaRand}^2\right] = \Expval\left[\norm{\noiseRand_0}^2\right] + \Expval\left[\norm{\zetaRand}^2\right] - 2 \nu \Expval\left[\norm{\noiseRand_0}^2\right]^{1/2}\Expval\left[\norm{\zetaRand}^2\right]^{1/2} ,
\]
where $\nu \in [-1,1]$ is the correlation coefficient between $\noiseRand_0$ and $\zetaRand$. Minimizing the above expression with respect to $\Expval\left[\norm{\zetaRand}^2\right]$ and $\nu$ gives that the optimal values are, respectively, $\Expval\left[\norm{\noiseRand_0}^2\right]$ and $1$. Clearly, the noise $\zetaRand$ ought to have the same variance as $\noiseRand_0$ and has correlation $1$ with it.

\section{Properties of $\tilde\genericop\tilde\convToShear$}
\label{appendix:orthogonalcompl}

Recall that $\tilde\genericop\tilde\convToShear := \tildeConvToShearTransp\tilde\covmatrNoise^{-1/2}\tilde\convToShear$ with the choice made in \eqref{eq:genericop}. This is obviously a symmetric positive semidefinite matrix, hence \ref{assum:3} follows. Observe also that:
\begin{equation}
    \image \tilde\genericop = \image \tildeConvToShearTransp\tilde\covmatrNoise^{-1/2} = \image \tildeConvToShearTransp,
\end{equation}
since $\tilde\covmatrNoise$ is invertible. On the other hand,
\begin{equation}
    \orthcomplBA = (\ker \tildeConvToShearTransp\tilde\covmatrNoise^{-1/2}\tilde\convToShear)^\perp = (\ker \tildeConvToShearTransp\tilde\convToShear)^\perp = (\ker \tilde\convToShear)^\perp,
\end{equation}
since $\tilde\covmatrNoise$ is invertible. As $\image \tildeConvToShearTransp = (\ker \tilde\convToShear)^\perp$, we conclude that \ref{assum:3'} holds.

\section{Proof of~\eqref{eq:almostidentity}}
\label{appendix:almostidentity}
We embark from \eqref{eq:deftildeAtildeSigma} to see that
\[
\convToShear^\top\convToShear = \Real(\convToShear^*\convToShear) ,
\]
where $\convToShear^*$ is the adjoint matrix of $\convToShear$. Using now \eqref{eq:defconvtoshear}, we get
\[
\convToShear^*\convToShear = \fouriermatrHerm|\fourierConvToShear|^2\fouriermatr .
\]
From \eqref{eq:deffourierconvtoshear}, it follows that $|\fourierConvToShear|^2$ is a diagonal matrix with
\[
|\fourierConvToShear[0,0]| = 0 \qand |\fourierConvToShear[i,i]| = 1, \forall i \geq 1.
\]
This entails that
\[
\convToShear^*\convToShear = \fouriermatrHerm\fouriermatr - \fouriermatr[1,\, :]^\top\fouriermatr[1,\, :] ,
\]
where $\fouriermatr[1,\, :] = $ is the first row of $\fouriermatr$. Let $\mathds{1}$ be the column vector of ones of size $K^2$. Since $\fouriermatr[1,\, :] := \mathds{1}^\top/K$ (the zero frequency component), and $\fouriermatr$ is unitary, we get that
\[
\convToShear^*\convToShear = \BBI - \mathds{1}\mathds{1}^\top/K^2 = \BBI - 1/K^2 .
\]

\section{Computing $\lambda_{\max}$ and $\lambda_{\min}$}
\label{appendix:powitlambdamin}

$\lambda_{\max}$ can be easily computed from the power iteration applied to $\tildeConvToShearTransp\tilde\covmatrNoise^{-1/2}\tilde\convToShear$ starting from a random uniform vector. We now turn to computing the smallest nonzero eigenvalue of $\tildeConvToShearTransp\tilde\covmatrNoise^{-1/2}\tilde\convToShear$. From Appendix~\ref{appendix:orthogonalcompl}, we have $\tildeConvToShearTransp\tilde\covmatrNoise^{-1/2}\tilde\convToShear = \ker \convToShear = \image \mathds{1}$, whose orthogonal is the subspace of zero-mean vectors. Let $\lambda_1=0< \lambda_2=\lambda_{\min} \leq \cdots \leq \lambda_{K^{^2}} = \lambda_{\max}$, be the eigenvalues of $\tildeConvToShearTransp\tilde\covmatrNoise^{-1/2}\tilde\convToShear$ sorted in ascending order. Denote $\bu_i$ the corresponding unit norm eigenvectors, where $\bu_1 = \mathds{1}^\top/K$. Define the matrix $\BBC := \lambda_{\max}\BBI - \tildeConvToShearTransp\tilde\covmatrNoise^{-1/2}\tilde\convToShear$, whose eigenvalues are $\lambda_{\max} - \lambda_i$, and eigenvectors are $\bu_i$. The key observation is that $\BBC$ is again symmetric positive semidefinite and, orthogonally to $\image \mathds{1}$, the largest eigenvalue of $\BBC$ is $\lambda_{\max} - \lambda_{\min}$.

The idea now is to apply the power iteration to $\BBC$ with an appropriate initial vector. For instance, we draw a vector $\bx_0$ form the uniform distribution (or any distribution which has a density with respect to the Lebesgue measure), in which case we can write
\[
\bx_0 = \sum_{i=1}^{K^2} c_i \bu_i, 
\]
where, for all $i$, $c_i \neq 0$ with probability one. We then project $\bx_0$ on $(\ker \convToShear)^\perp = (\image \mathds{1})^\perp$, that is,
\[
\bar{\bx}_0 = \bx_0 - \frac{1}{K^2}\sum_{i=1}^{K^2} \bx_0[i] .
\]
Since $(\ker \convToShear)^\perp$ is the span of $(\bu_i)_{i \geq 2}$, we have 
\[
\bar{\bx}_0 = \sum_{i=2}^{K^2} c_i \bu_i .
\]
Hence,
\[
\BBC^k \bar{\bx}_0 = c_2 (\lambda_{\max} - \lambda_{\min})^k \bu_2 + \sum_{i=3}^{K^2-1} c_i (\lambda_{\max} - \lambda_i)^k \bu_i .
\]
Applying the power iteration to $\BBC$ with initial vector $\bar{\bx}_0$ allows computing $\lambda_{\max} - \lambda_{\min}$. Plugging the value of $\lambda_{\max}$ (computed using, \eg, power iteration) gives an estimate of $\lambda_{\min}$.

\section{PnPMass and the FBS algorithm}
\label{appendix:pnp_fb}

We assume the existence of a proper convex, lower semicontinuous function $\regfun$ such that
\begin{equation}
    \deepmodelTrained(\convmap',\, \stepsize) = \prox_{\stepsize\regfun}(\convmap'),
\label{eq:deepmodelasprox}
\end{equation}
where the proximal operator is defined as
\begin{equation}
    \prox_{\stepsize\regfun}: \convmap' \mapsto \argmin_{\convmap''} \frac12 \normtwo{\convmap'' - \convmap'}^2 + \stepsize\regfun(\convmap'').
\end{equation}
Under this assumption, the fixed-point iteration \eqref{eq:fixpointoperator_noiseaware} can be written as
\begin{equation}
    \fixedpointiter{\tilde\shearmap} (\convmap',\, \stepsize) := \prox_{\stepsize\regfun}\left[
        \convmap' - \stepsize\nabla\!\datafidelity_{\tilde\shearmap}(\tilde\convToShear\,\cdot)(\convmap')
    \right],
\label{eq:fixpointoperator_prox}
\end{equation}
where we have considered the following data fidelity term
\begin{equation}
    \datafidelity_{\tilde\shearmap}(\tilde\convToShear\convmap') := \frac12\norm{\tilde\shearmap - \tilde\convToShear\convmap'}_{\tilde\covmatrNoise^{-1/2}}^2.
\label{eq:datafidelity}
\end{equation}
The operator \eqref{eq:fixpointoperator_prox} is nothing but the fixed-point operator of the well-known FBS algorithm.
It consists of a forward step---gradient descent with respect to the data fidelity term $\datafidelity_{\tilde\shearmap}(\tilde\convToShear\,\cdot)$---and a backward step---proximal mapping with respect to the regularization term $g$. This is the reason why PnP-type algorithms with an FBS structure are usually referred to as PnP-FBS in the literature \citep{Ryu2019,Ebner2024}.

It is well-known that if the step-size satisfies $0 < \tau < 2/\bignorm{\tilde\convToShear^\top\tilde\covmatrNoise^{-1/2} \tilde\convToShear}$, then the sequence of iterates of FBS with operator \eqref{eq:fixpointoperator_prox} converges to a minimizer $\convmapEstimate$ of the following minimization problem---similar in form to~\eqref{eq:minprob}:
\begin{equation}
    \convmapEstimate \in \argmin_{\convmap'} \datafidelity_{\tilde\shearmap}(\tilde\convToShear\convmap') + \regfun(\convmap'),
\label{eq:minprob2}
\end{equation}
with the proviso that the set of minimizers is non-empty. 

The FBS has been widely used for astronomical data analysis; see the comprehensive treatment in \citep{Starck2015}. The iterative Wiener filtering algorithm \citep{Bobin2012} for mass mapping (see Sect.~\ref{subsubsec:background_relatedwork_classical}), is also an instance of the FBS algorithm.

While some proximal operators can be seen as denoisers, the converse is clearly not true. Thus, there is no reason that a deep-learning based denoiser $\deepmodelTrained(\cdot,\, \stepsize)$ to be interpreted as a proximal operator. At best, the denoiser network can be constrained during training to be firmly non-expansive, that is, to be the resolvent of a maximally monotone set-valued operator as proposed in \citep{Terris2020}. While this point of view brings interpretability, it is however not clear it leads to better denoising results. Consequently, PnP algorithms such as PnPMass are not expected to solve any optimization problem in general, hence the fixed-point perspective that we advocated in this paper.

It is worth recalling again (see Sect.~\ref{subsec:pnp_operator}) that the data fidelity term defined in ~\eqref{eq:datafidelity} does not correspond to the negative log-likelihood typically used in Bayesian frameworks: the Mahalanobis norm would be defined with respect to $\tilde\covmatrNoise^{-1}$ rather than $\tilde\covmatrNoise^{-1/2}$. Furthermore, the regularization function $\regfun$ does not necessarily correspond to a negative log-prior; see \citet{Starck2013} for a thorough discussion. As a result, $\convmapEstimate$ should not be interpreted as a maximum a posteriori (MAP) estimate, even if~\eqref{eq:deepmodelasprox} is satisfied.

\onecolumn

\section{Additional plots}
\label{appendix:plots}

\vspace{-10pt}
\begin{figure*}[h]
    \centering
    \begin{subfigure}{0.48\textwidth}
        \centering
        \includegraphics[width=0.9\linewidth]{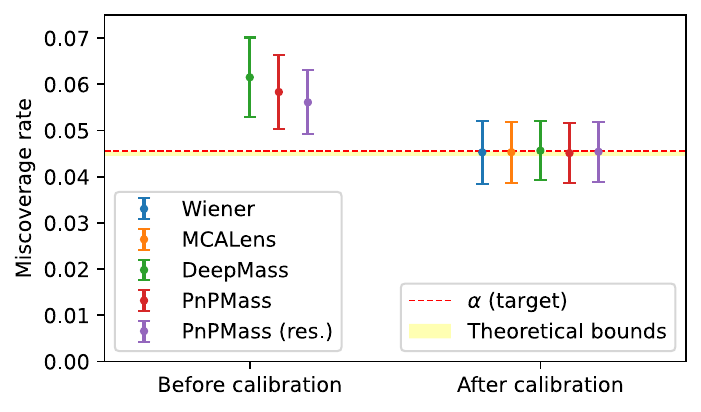}
        \vspace{-5pt}
        \caption{Empirical miscoverage rate over active pixels, before and after calibration (using $1\,024$ calibration examples). Error bars indicate the mean and standard deviation across the $512$ test examples from the $\kappa$TNG dataset. As in Fig.~\ref{fig:summary}, Wiener and MCALens results are only displayed after calibration. The theoretical bounds (in yellow) come from \eqref{eq:miscoveragerate_cqr}.}
        \label{fig:miscoveragerate}
    \end{subfigure}
    \hfill
    \begin{subfigure}{0.48\textwidth}
        \centering
        \includegraphics[trim=0 0 0 7pt, clip, width=0.9\linewidth]{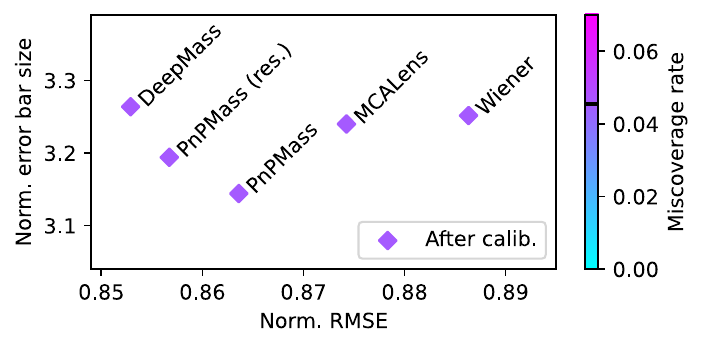}
        \vspace{-5pt}
        \caption{
            Mean length of prediction intervals plotted against reconstruction accuracy, after calibration. Unlike Fig.~\ref{fig:summary}, all methods have been initialized with zero-valued error bars (no order-2 moment network). We observe that ignoring the pre-calibration step with order-2 moment networks result in larger error bars for PnPMass and DeepMass, by comparison with Fig.~\ref{fig:summary}. Furthermore, DeepMass now produces larger error bars than Wiener or MCALens, despite being significantly more accurate.
        }
        \label{fig:summary_zeroinit}
    \end{subfigure}
    \caption{Left: Miscoverage rate before and after calibration. Right: Zero-initialization error bars.}
    \label{fig:additional_plots}
\end{figure*}

\section{Visual representations}
\label{appendix:visual_representations}

\vspace{-5pt}
\begin{figure*}[h]
    \centering
    \begin{subfigure}[b]{0.28\textwidth}
        \centering
        \includegraphics[height=0.15\textheight]{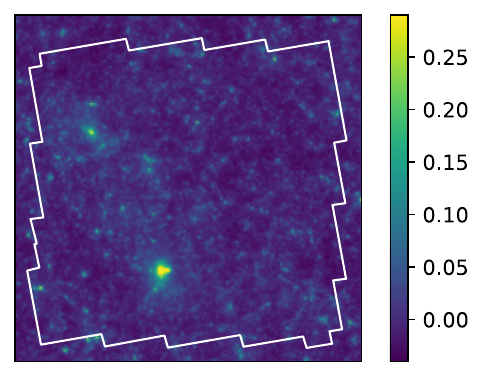}
        \label{subfig:convshear_conv}
    \end{subfigure}
    \hspace{5pt}
    \begin{subfigure}[b]{0.28\textwidth}
        \centering
        \includegraphics[height=0.15\textheight]{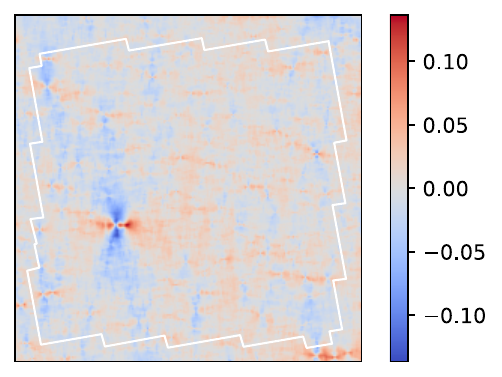}
        \label{subfig:convshear_shear_re}
    \end{subfigure}
    \hspace{5pt}
    \begin{subfigure}[b]{0.28\textwidth}
        \centering
        \includegraphics[height=0.15\textheight]{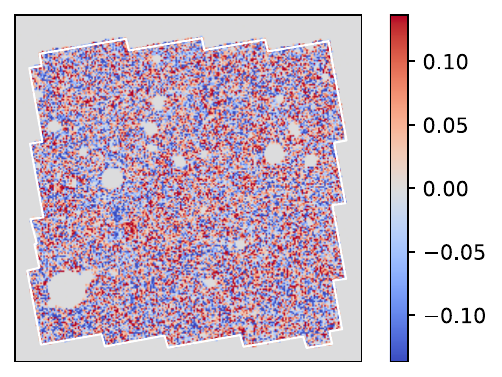}
        \label{subfig:convshear_shear_im}
    \end{subfigure}
    \caption{Example of mock convergence map obtained from $\kappa$TNG simulations (left), and the real part of the corresponding clean (middle) and noisy (right) shear map. The COSMOS boundaries are delimited in white.
    }
    \label{fig:convshear}
\end{figure*}

\vspace{-5pt}

\begin{figure*}[h]
    \centering
    \small
    \vspace{-8pt}
    \begin{tikzpicture}


        \node (mmgan) {\includegraphics[height=0.21\textwidth, trim={0 0 50pt 0}, clip]{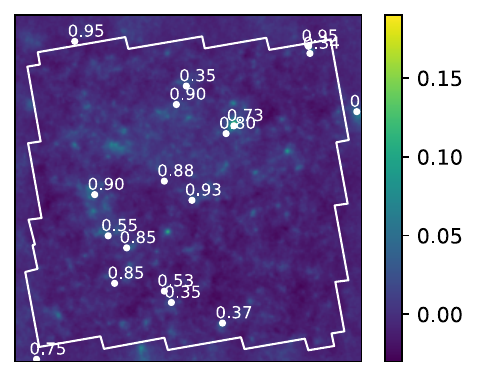}};
        \node[anchor=west, xshift=-5pt] (deepmass) at (mmgan.east) {\includegraphics[height=0.21\textwidth, trim={0 0 50pt 0}, clip]{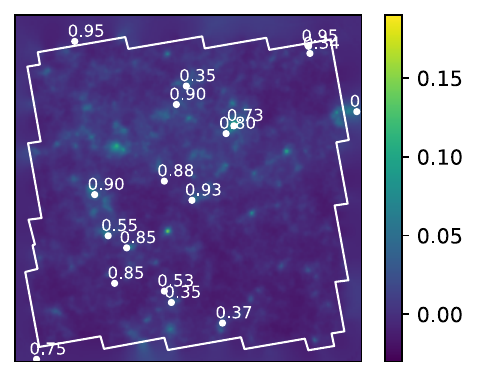}};
        \node[anchor=west, xshift=-5pt] (pnpmass) at (deepmass.east) {\includegraphics[height=0.21\textwidth, trim={0 0 50pt 0}, clip]{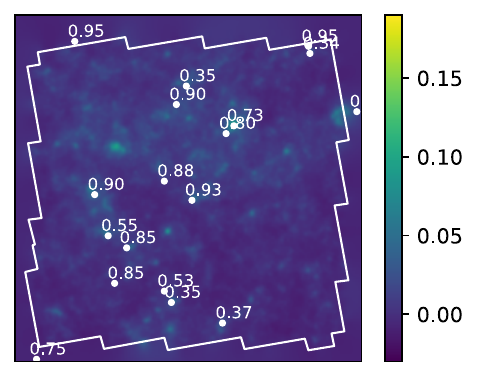}};
        \node[anchor=west, xshift=-5pt] (respnpmass) at (pnpmass.east) {\includegraphics[height=0.21\textwidth]{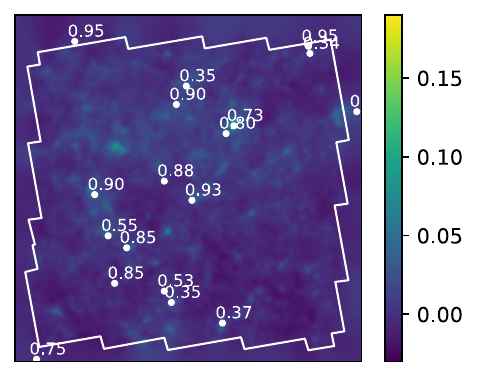}};


        \node[anchor=north west, yshift=9pt] (mmgan_std) at (mmgan.south west) {\includegraphics[height=0.206\textwidth, trim={0 4.8pt 52pt 4.8pt}, clip]{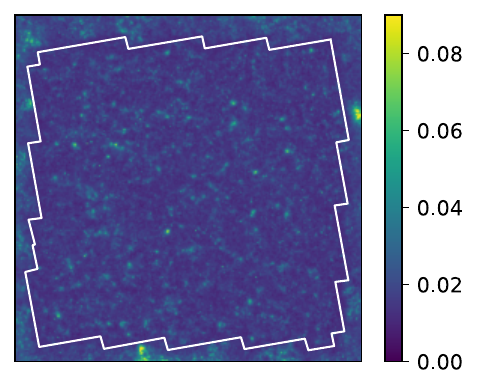}};
        \node[anchor=north west, yshift=9pt] (deepmass_std) at (deepmass.south west) {\includegraphics[height=0.206\textwidth, trim={0 4.8pt 52pt 4.8pt}, clip]{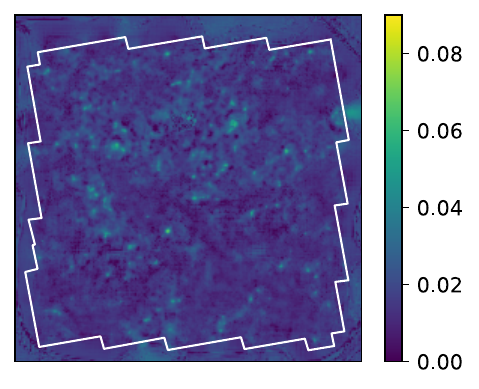}};
        \node[anchor=north west, yshift=9pt] (pnpmass_std) at (pnpmass.south west) {\includegraphics[height=0.206\textwidth, trim={0 4.8pt 52pt 4.8pt}, clip]{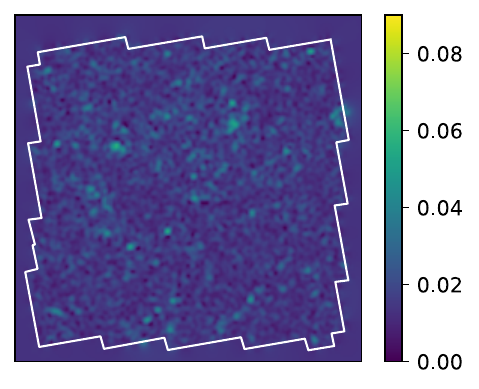}};
        \node[anchor=north west, yshift=9pt] (respnpmass_std) at (respnpmass.south west) {\includegraphics[height=0.206\textwidth, trim={0 4.8pt 0 4.8pt}, clip]{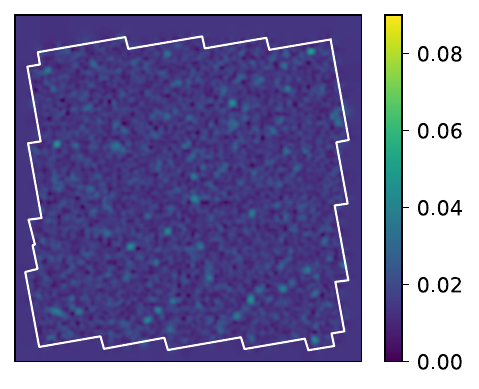}};


        \node (mmgan caption) at (mmgan.north) {MMGAN};
        \node (deepmass caption) at (deepmass.north) {DeepMass};
        \node (pnpmass caption) at (pnpmass.north) {\textbf{PnPMass}};
        \node[xshift=-12pt] (respnpmass caption) at (respnpmass.north) {\textbf{PnPMass (res.)}};


        \node[xshift=-2pt, rotate=90] (pointestimate 1) at (mmgan.west) {$\convmapEstimate$};
        \node[xshift=-2pt, rotate=90] (stdestimate 1) at (mmgan_std.west) {$\stdmapEstimate$};

        \end{tikzpicture}
    \vspace{-8pt}
    \caption{Reconstruction of the COSMOS field, using the same settings as in Fig.~\ref{fig:visual_representations_cosmos}.
    Top row: point estimates; the white scatter plot indicates the location of known x-ray clusters \citep{FinoguenovXmmNewton2007}. Bottom row: uncertainty maps, obtained with order-2 moment networks for PnPMass and DeepMass, and by computing the per-pixel standard deviation over $32$ samples for MMGAN.}
    \label{fig:visual_representations_cosmos_bis}
\end{figure*}

\begin{figure*}[h]
    \centering
    \small
    \vspace{-8pt}
    \begin{tikzpicture}


        \node (wiener) {\includegraphics[height=0.188\textwidth, trim={0 0 50pt 0}, clip]{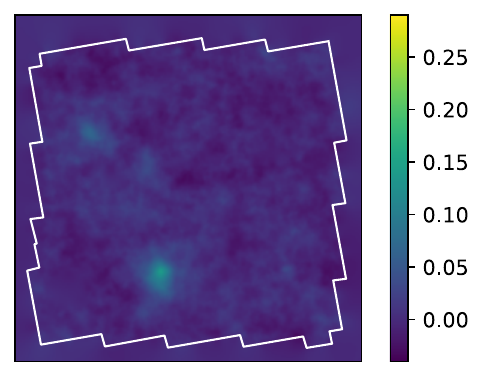}};
        \node[anchor=west, xshift=-10pt] (mcalens) at (wiener.east) {\includegraphics[height=0.188\textwidth, trim={0 0 50pt 0}, clip]{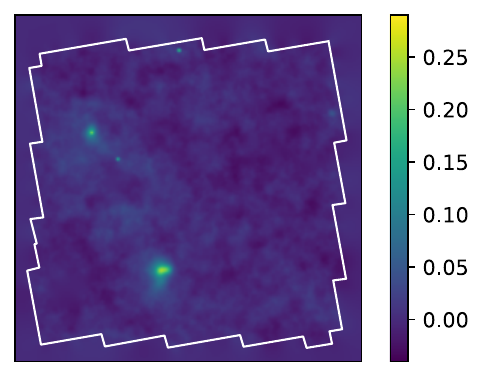}};
        \node[anchor=west, xshift=-10pt] (deepmass) at (mcalens.east) {\includegraphics[height=0.188\textwidth, trim={0 0 50pt 0}, clip]{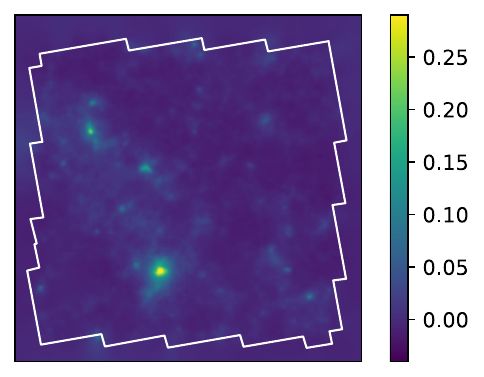}};
        \node[anchor=west, xshift=-10pt] (pnpmass) at (deepmass.east) {\includegraphics[height=0.188\textwidth, trim={0 0 50pt 0}, clip]{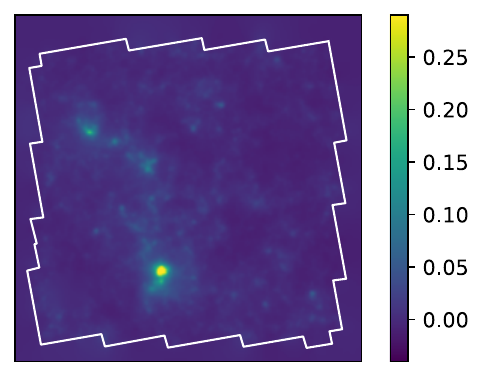}};
        \node[anchor=west, xshift=-10pt] (respnpmass) at (pnpmass.east) {\includegraphics[height=0.188\textwidth]{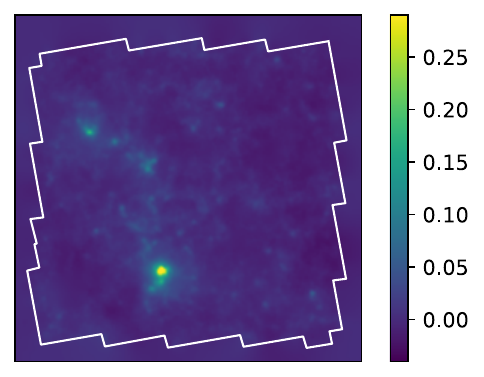}};


        \node[anchor=north west, yshift=9pt] (wiener_std) at (wiener.south west) {\includegraphics[height=0.188\textwidth, trim={0 4.8pt 52pt 4.8pt}, clip]{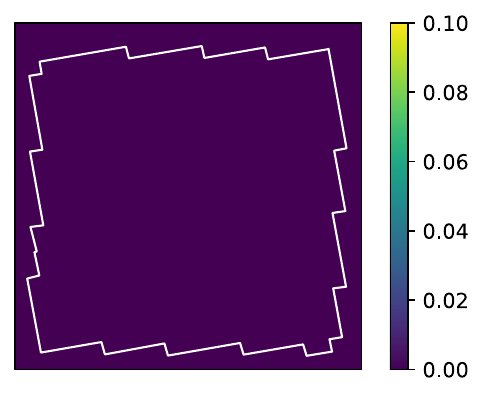}};
        \node[anchor=north west, yshift=9pt] (mcalens_std) at (mcalens.south west) {\includegraphics[height=0.188\textwidth, trim={0 4.8pt 52pt 4.8pt}, clip]{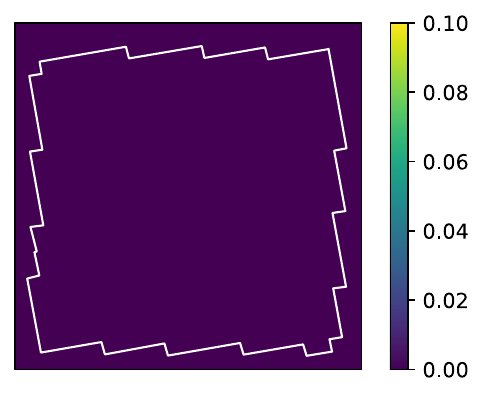}};
        \node[anchor=north west, yshift=9pt] (deepmass_std) at (deepmass.south west) {\includegraphics[height=0.188\textwidth, trim={0 4.8pt 52pt 4.8pt}, clip]{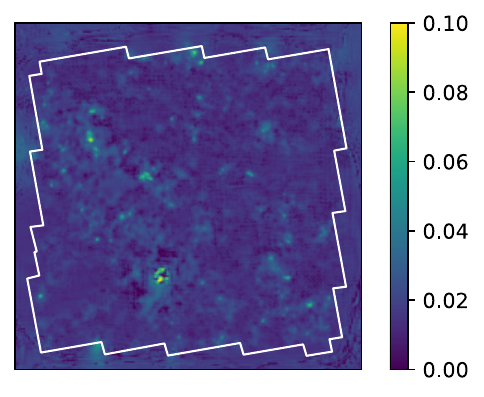}};
        \node[anchor=north west, yshift=9pt] (pnpmass_std) at (pnpmass.south west) {\includegraphics[height=0.188\textwidth, trim={0 4.8pt 52pt 4.8pt}, clip]{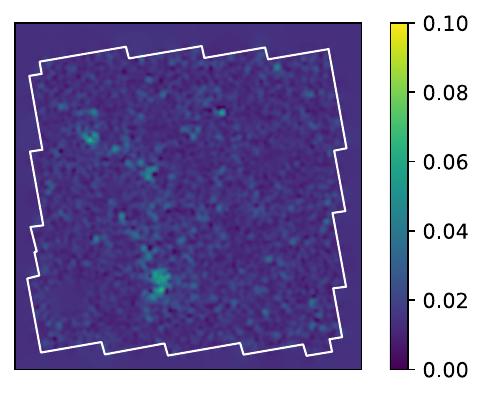}};
        \node[anchor=north west, yshift=9pt] (respnpmass_std) at (respnpmass.south west) {\includegraphics[height=0.188\textwidth, trim={0 4.8pt 0 4.8pt}, clip]{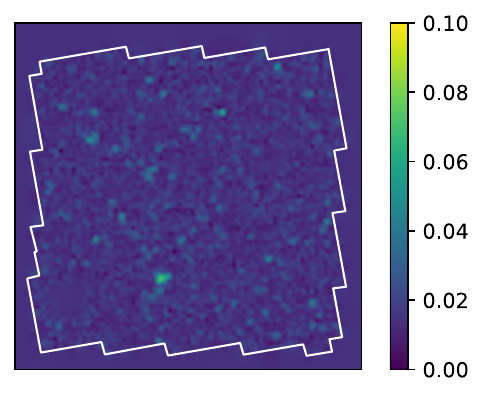}};


        \node[anchor=north west, yshift=9pt] (wiener_low) at (wiener_std.south west) {\includegraphics[height=0.188\textwidth, trim={0 0 52pt 0}, clip]{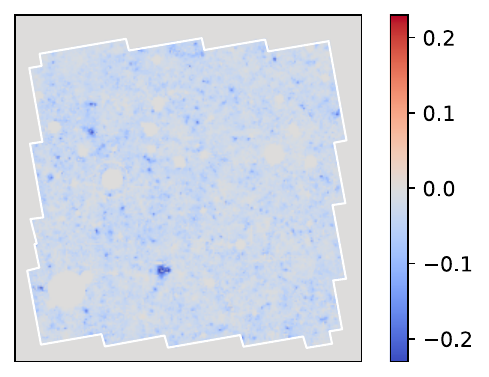}};
        \node[anchor=north west, yshift=9pt] (mcalens_low) at (mcalens_std.south west) {\includegraphics[height=0.188\textwidth, trim={0 0 52pt 0}, clip]{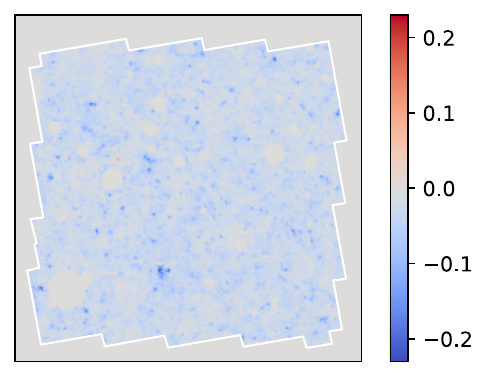}};
        \node[anchor=north west, yshift=9pt] (deepmass_low) at (deepmass_std.south west) {\includegraphics[height=0.188\textwidth, trim={0 0 52pt 0}, clip]{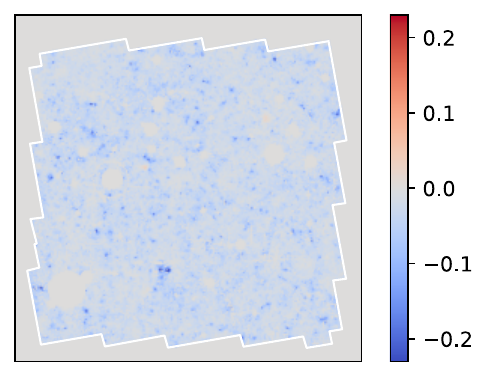}};
        \node[anchor=north west, yshift=9pt] (pnpmass_low) at (pnpmass_std.south west) {\includegraphics[height=0.188\textwidth, trim={0 0 52pt 0}, clip]{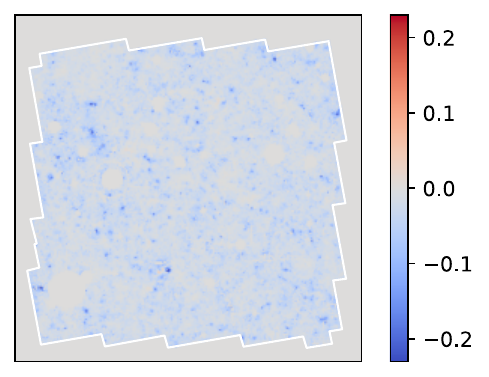}};
        \node[anchor=north west, yshift=9pt] (respnpmass_low) at (respnpmass_std.south west) {\includegraphics[height=0.188\textwidth]{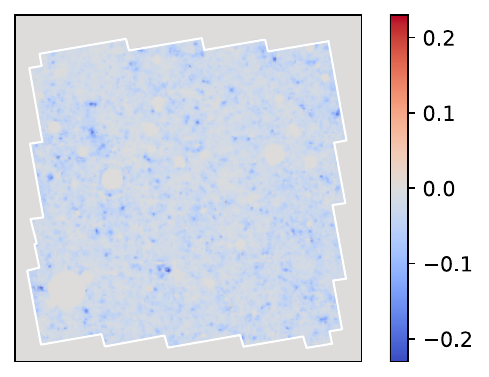}};

        \node[anchor=north west, yshift=9pt] (wiener_high) at (wiener_low.south west) {\includegraphics[height=0.188\textwidth, trim={0 0 52pt 0}, clip]{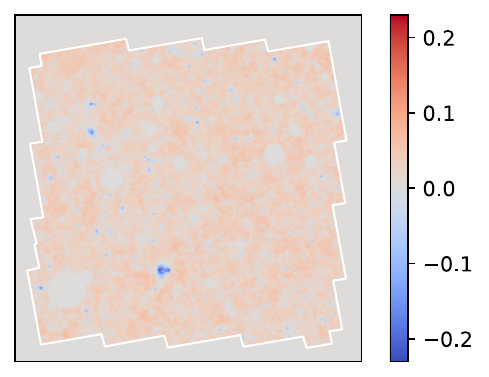}};
        \node[anchor=north west, yshift=9pt] (mcalens_high) at (mcalens_low.south west) {\includegraphics[height=0.188\textwidth, trim={0 0 52pt 0}, clip]{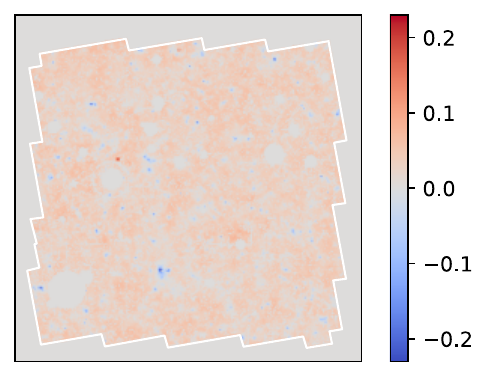}};
        \node[anchor=north west, yshift=9pt] (deepmass_high) at (deepmass_low.south west) {\includegraphics[height=0.188\textwidth, trim={0 0 52pt 0}, clip]{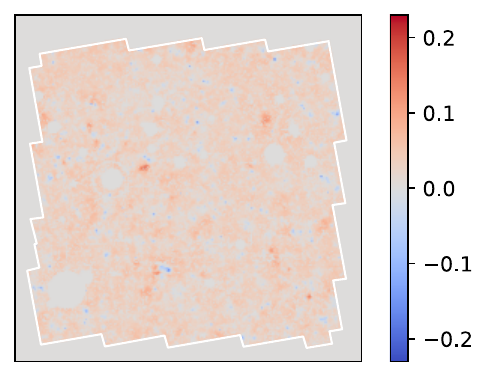}};
        \node[anchor=north west, yshift=9pt] (pnpmass_high) at (pnpmass_low.south west) {\includegraphics[height=0.188\textwidth, trim={0 0 52pt 0}, clip]{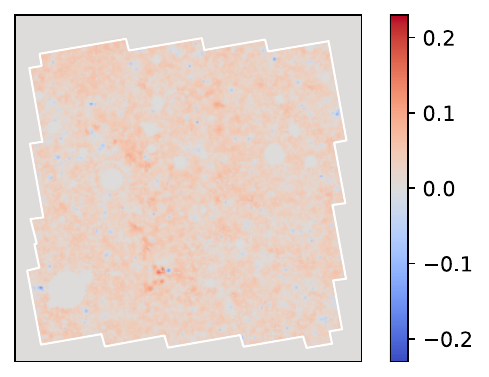}};
        \node[anchor=north west, yshift=9pt] (respnpmass_high) at (respnpmass_low.south west) {\includegraphics[height=0.188\textwidth]{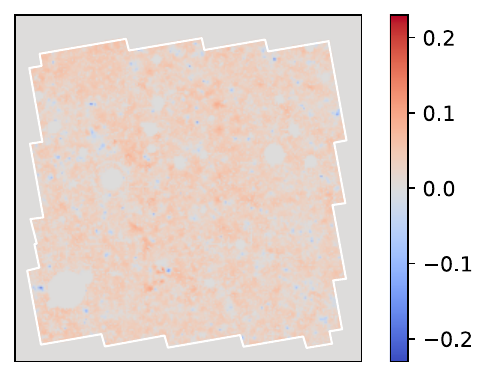}};


        \node (wiener caption) at (wiener.north) {Wiener$^\ast$};
        \node (mcalens caption) at (mcalens.north) {MCALens$^\ast$};
        \node (deepmass caption) at (deepmass.north) {DeepMass};
        \node (pnpmass caption) at (pnpmass.north) {\textbf{PnPMass}};
        \node[xshift=-12pt] (respnpmass caption) at (respnpmass.north) {\textbf{PnPMass (res.)}};


        \node[xshift=-2pt, rotate=90] (pointestimate 1) at (wiener.west) {$\convmapEstimate$};
        \node[xshift=-2pt, rotate=90] (stdestimate 1) at (wiener_std.west) {$\stdmapEstimate$};
        \node[xshift=-2pt, rotate=90] (nocalib low 1) at (wiener_low.west) {$\convmapEstimateLowCQR - \convmap$};
        \node[xshift=-2pt, rotate=90] (nocalib high 1) at (wiener_high.west) {$\convmapEstimateHighCQR - \convmap$};

        \end{tikzpicture}
    \vspace{-8pt}
    \caption{Example of mass-mapping reconstructions on the $\kappa$TNG test dataset, for each mass-mapping method, run with the same parameters as in Table~\ref{table:rmse}. The corresponding ground truth convergence map $\convmap$ and noisy input $\shearmap$ are displayed in Fig.~\ref{fig:convshear}. Top row: Point estimates, $\convmapEstimate$. Second row: Standard deviation estimates, $\stdmapEstimate$, obtained with order-2 moment networks (used for pre-calibration error bars). Bottom rows: Lower and upper bounds, $\convmapEstimateLowCQR$ and $\convmapEstimateHighCQR$, after calibration with CQR. For visualization purpose, the bounds have been centered with respect to the ground truth $\convmap$. $^\ast$Zero-sized error bars before calibration.}
    \label{fig:visual_representations}
\end{figure*}
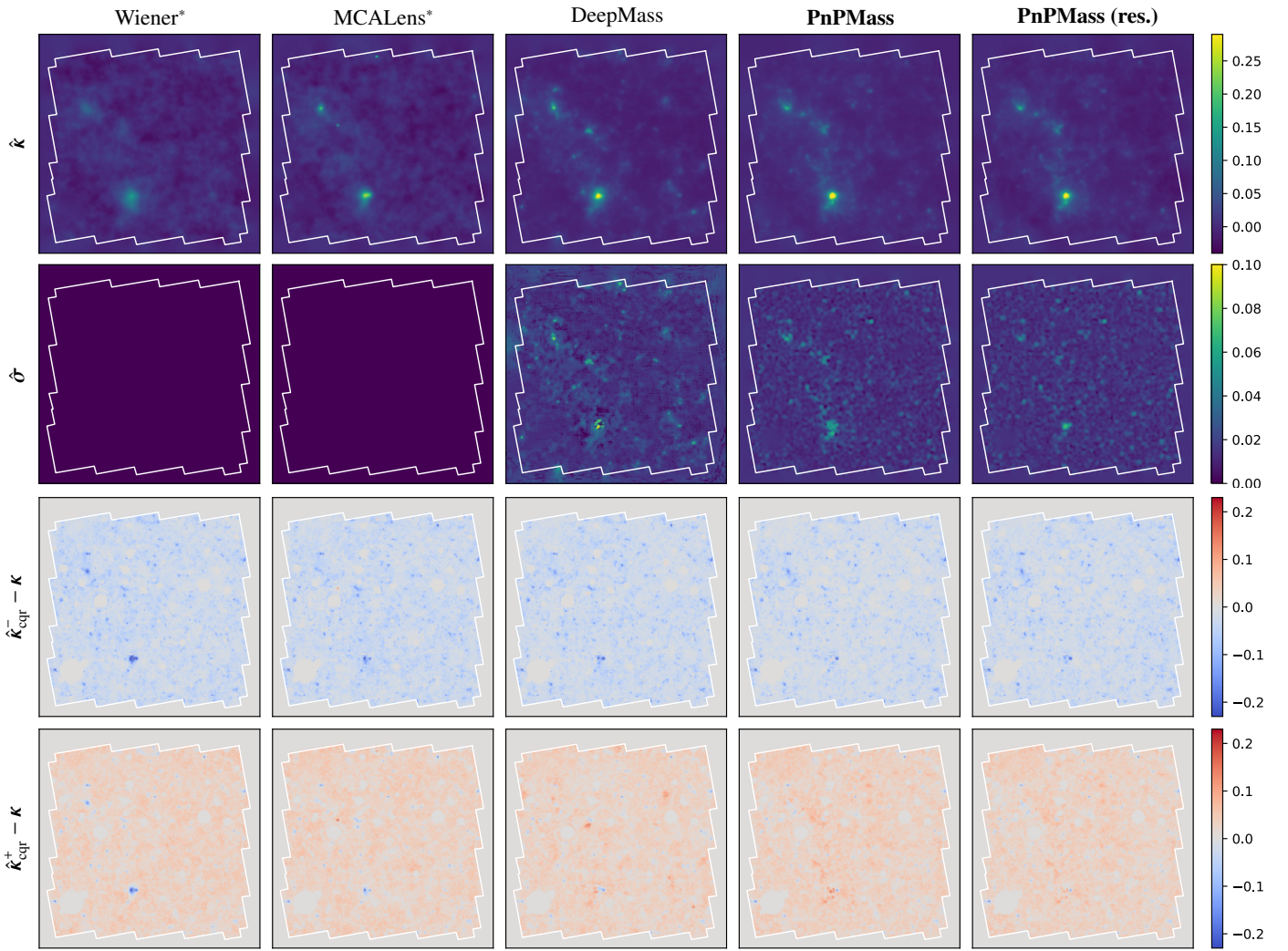

\end{appendix}

\end{document}

%% file: inputs/mycommands.tex
\DeclareMathOperator*{\argmin}{argmin}

\DeclareMathOperator{\Real}{Re}
\DeclareMathOperator{\Imag}{Im}

\DeclareMathOperator{\prox}{prox}

\DeclareMathOperator{\image}{ran}

\newcommand{\cmark}{\ding{51}}
\newcommand{\xmark}{\ding{55}}

\newcommand{\mathC}{\mathbb{C}}
\newcommand{\mathE}{\mathbb{E}}

\newcommand{\mathN}{\mathbb{N}}
\newcommand{\mathP}{\mathbb{P}}

\newcommand{\mathR}{\mathbb{R}}

\newcommand{\mathV}{\mathbb{V}}

\newcommand{\MD}{\mathcal D}

\newcommand{\MF}{\mathcal F}
\newcommand{\MG}{\mathcal G}
\newcommand{\MH}{\mathcal H}

\newcommand{\MM}{\mathcal M}
\newcommand{\MN}{\mathcal N}

\newcommand{\MP}{\mathcal P}

\newcommand{\MU}{\mathcal U}
\newcommand{\MV}{\mathcal V}
\newcommand{\MW}{\mathcal W}

\newcommand{\bn}{\boldsymbol n}

\newcommand{\br}{\boldsymbol r}

\newcommand{\bu}{\boldsymbol u}
\newcommand{\bv}{\boldsymbol v}
\newcommand{\bw}{\boldsymbol w}
\newcommand{\bx}{\boldsymbol x}

\newcommand{\BBA}{\mathbf A}
\newcommand{\BBB}{\mathbf B}
\newcommand{\BBC}{\mathbf C}

\newcommand{\BBF}{\mathbf F}
\newcommand{\BBI}{\mathbf I}

\newcommand{\BBM}{\mathbf M}

\newcommand{\BBP}{\mathbf P}

\newcommand{\SFK}{\mathsf{K}}

\newcommand{\SFN}{\mathsf{N}}

\newcommand{\SFR}{\mathsf{R}}

\newcommand{\SFT}{\mathsf{T}}

\newcommand{\SFZ}{\mathsf{Z}}

\newcommand{\SFGamma}{\mathsf{\Gamma}}

\newcommand{\bgamma}{\boldsymbol \gamma}

\newcommand{\bzeta}{\boldsymbol \zeta}

\newcommand{\bkappa}{\boldsymbol \kappa}

\newcommand{\btheta}{\boldsymbol \theta}

\newcommand{\blambda}{\boldsymbol \lambda}
\newcommand{\bsigma}{\boldsymbol \sigma}

\newcommand{\bomega}{\boldsymbol \omega}

\newcommand{\BSigma}{\boldsymbol \Sigma}

\newcommand{\zetaRand}{\boldsymbol \SFZ}

\newcommand{\bzero}{\boldsymbol 0}

\newcommand{\red}[1]{{\color{red}{#1}}}

\newcommand{\range}[2]{\left\{#1, \ldots, #2\right\}}
\newcommand{\smallrange}[2]{\{#1, \ldots, #2\}}

\newcommand{\oneto}[1]{\range{0}{#1}}
\newcommand{\smalloneto}[1]{\smallrange{1}{#1}}

\newcommand{\interval}[2]{\left[#1,\, #2\right]}

\newcommand{\intervalexclr}[2]{\left[#1,\, #2\right[}
\newcommand{\intervalexcllr}[2]{\left]#1,\, #2\right[}

\newcommand{\zeroexcloneexcl}{\intervalexcllr{0}{1}}

\newcommand{\tuple}[2]{\{#1,\, #2\}}

\newcommand{\Proba}{\mathP}
\newcommand{\Expval}{\mathE}
\newcommand{\Var}{\mathV}

\newcommand{\condproba}[2]{\Proba\left\{ #1 \,\middle|\, #2 \right\}}

\newcommand{\condexpval}[2]{\Expval\left[ #1 \;\middle|\; #2 \right]}
\newcommand{\bigcondexpval}[2]{\Expval\bigl[ #1 \,|\, #2 \bigr]}
\newcommand{\Bigcondexpval}[2]{\Expval\Bigl[ #1 \bigm| #2 \Bigr]}

\newcommand{\condvar}[2]{\Var\left[ #1 \;\middle|\; #2 \right]}
\newcommand{\bigcondvar}[2]{\Var\bigl[ #1 \,|\, #2 \bigr]}
\newcommand{\Bigcondvar}[2]{\Var\Bigl[ #1 \bigm| #2 \Bigr]}

\renewcommand{\arraystretch}{1.5} 

\renewcommand{\texttt}[1]{%
	\begingroup
	\ttfamily
	\begingroup\lccode`~=`/\lowercase{\endgroup\def~}{/\discretionary{}{}{}}%
	\begingroup\lccode`~=`[\lowercase{\endgroup\def~}{[\discretionary{}{}{}}%
	\begingroup\lccode`~=`.\lowercase{\endgroup\def~}{.\discretionary{}{}{}}%
	\catcode`/=\active\catcode`[=\active\catcode`.=\active
	\scantokens{#1\noexpand}%
	\endgroup
} 

\newcommand{\norm}[1]{\left\|#1\right\|}

\newcommand{\normtwo}[1]{\norm{#1}_2}

\newcommand{\bignorm}[1]{\bigl\|#1\bigr\|}

\newcommand{\smallnorm}[1]{\|#1\|}

\newcommand{\qand}{\quad\mbox{and}\quad}
\newcommand{\qqand}{\qquad\mbox{and}\qquad}

\newcommand{\qwhere}{\quad\mbox{where}\quad}

\newcommand{\qwith}{\quad\mbox{with}\quad}

\newcommand{\todo}[1][]{%
  	{\red{\ifthenelse{\equal{#1}{}}%
    	{\textbf{[TODO]}}%
    	{\textbf{[TODO: #1]}}%
  	}}\xspace
}

\makeatletter
\DeclareRobustCommand\onedot{\futurelet\@let@token\@onedot}
\def\@onedot{\ifx\@let@token.\else.\null\fi\xspace}

\def\eg{e.g\onedot} 
\def\ie{i.e\onedot}

\def\wrt{w.r.t\onedot}

\makeatother

%% file: inputs/mycommands_specific.tex
\DeclareMathOperator*{\herm}{\ast}
\DeclareMathOperator*{\ellip}{ell}

\DeclareMathOperator*{\deepmass}{dm}
\DeclareMathOperator*{\deepposterior}{dp}

\DeclareMathOperator*{\cqr}{cqr}

\DeclareMathOperator*{\gauss}{g}
\DeclareMathOperator*{\nongauss}{ng}
\DeclareMathOperator*{\iterations}{it}

\DeclareMathOperator*{\maxval}{max}

\DeclareMathOperator{\forward}{\MF}

\DeclareMathOperator*{\baryon}{b}
\DeclareMathOperator*{\matter}{m}
\DeclareMathOperator*{\spectralindex}{s}

\DeclareMathOperator{\km}{km}
\DeclareMathOperator{\s}{s}
\DeclareMathOperator{\Mpc}{Mpc}

\newcommand{\shearmap}{\bgamma}
\newcommand{\convmap}{\bkappa}
\newcommand{\varmap}{\bv}
\newcommand{\stdmap}{\bsigma}

\newcommand{\noise}{\bn}
\newcommand{\residual}{\br}
\newcommand{\residualRand}{\boldsymbol{\SFR}}

\newcommand{\stdval}{\sigma}

\newcommand{\shearmapRand}{\boldsymbol{\SFGamma}}
\newcommand{\convmapRand}{\boldsymbol{\SFK}}
\newcommand{\noiseRand}{\boldsymbol{\SFN}}

\newcommand{\covmatr}{\BSigma}
\newcommand{\covmatrNoise}{\covmatr}

\newcommand{\fouriercovmatr}{\BBP}

\newcommand{\fouriermatr}{\BBF}
\newcommand{\fouriermatrHerm}{\fouriermatr^{\herm}}
\newcommand{\powerspectrum}{\fouriercovmatr_{\!\gauss}}

\newcommand{\convToShear}{\BBA}
\newcommand{\mask}{\BBM}
\newcommand{\fourierConvToShear}{\BBP}

\newcommand{\genericop}{\BBB}

\newcommand{\convToShearHerm}{\convToShear^{\!\herm}}
\newcommand{\tildeConvToShearTransp}{{\tilde\convToShear}^{\!\top}}

\newcommand{\convmapEstimate}{\hat\convmap}
\newcommand{\convmapEstimateLow}{\convmapEstimate^{-}}
\newcommand{\convmapEstimateHigh}{\convmapEstimate^{+}}
\newcommand{\convmapEstimateHighLow}{\convmapEstimate^{\pm}}

\newcommand{\convmapEstimateLowCQR}{\convmapEstimateLow_{\cqr}}
\newcommand{\convmapEstimateHighCQR}{\convmapEstimateHigh_{\cqr}}
\newcommand{\residualEstimate}{\hat\residual}
\newcommand{\residualEstimateRand}{\hat\residualRand}

\newcommand{\convmapEstimateRand}{\hat\convmapRand}
\newcommand{\convmapEstimateLowRand}{\convmapEstimateRand^{-}}
\newcommand{\convmapEstimateHighRand}{\convmapEstimateRand^{+}}

\newcommand{\convmapEstimateLowCQRRand}{\convmapEstimateLowRand_{\cqr}}
\newcommand{\convmapEstimateHighCQRRand}{\convmapEstimateHighRand_{\cqr}}

\newcommand{\convmapGauss}{\convmap^{\gauss}}
\newcommand{\convmapNongauss}{\convmap^{\nongauss}}

\newcommand{\shearmapLiftedNongauss}{\tilde\shearmap^{\nongauss}}

\newcommand{\varmapEstimate}{\hat\varmap}
\newcommand{\stdmapEstimate}{\hat\stdmap}

\newcommand{\regfun}{g}
\newcommand{\datafidelity}{f}
\newcommand{\pdf}{p}

\newcommand{\pnpmass}{\MP}
\newcommand{\pnpmassnoisy}{\pnpmass}
\newcommand{\iterativewiener}[1]{\MW_{\!#1}}

\newcommand{\forwardop}[1]{\forward_{\! #1}}
\newcommand{\fixedpointiter}[1]{\MH_{#1}}
\newcommand{\deepmodel}{\MD}
\newcommand{\deepmassmodel}{\MM}
\newcommand{\deepmodelbis}{\MG}
\newcommand{\deepmodelbisvar}{\MV}
\newcommand{\deepmodelParam}[1]{\deepmodel_{#1}}
\newcommand{\deepmassmodelParam}[1]{\deepmassmodel_{#1}}
\newcommand{\deepmodelbisParam}[1]{\deepmodelbis_{\! #1}}

\newcommand{\paramDeepmodel}{{\btheta}}
\newcommand{\paramDeepmodelbis}{{\bomega}}
\newcommand{\paramDeepmodelbisvar}{{\bomega}}
\newcommand{\paramDeepmassmodel}{{{\bw}}}

\newcommand{\deepmodelParameterized}{\deepmodelParam{\paramDeepmodel}}
\newcommand{\deepmassmodelParameterized}{\deepmassmodelParam{\paramDeepmassmodel}}
\newcommand{\deepmodelbisParameterized}{\deepmodelbisParam{\paramDeepmodelbis}}
\newcommand{\deepmodelbisvarParameterized}{\deepmodelbisParam{\paramDeepmodelbisvar}}
\newcommand{\deepmodelTrained}{\deepmodelParam{\hat\paramDeepmodel}}
\newcommand{\deepmassmodelTrained}{\deepmassmodelParam{\hat\paramDeepmassmodel}}
\newcommand{\deepmodelbisTrained}{\deepmodelbisParam{\hat\paramDeepmodelbis}}

\newcommand{\cdfg}{\Phi}

\newcommand{\imgsize}{K}
\newcommand{\ngalperpix}[1]{N_{#1}}
\newcommand{\sizeTrainingset}{n}
\newcommand{\sizeCalibrationset}{m}
\newcommand{\niter}{N_{\iterations}}

\newcommand{\stepsize}{\tau}
\newcommand{\stepsizeRand}{\SFT}

\newcommand{\stepsizeMax}{\stepsize_{\!\maxval}}

\newcommand{\multfactbounds}{\mu}

\newcommand{\errbarsize}{\epsilon}
\newcommand{\lipschitzcst}{\beta}
\newcommand{\calibparam}{\lambda}
\newcommand{\calibparamVec}{\blambda}

\newcommand{\datasetMM}[2]{(\shearmap_i,\, \convmap_i)_{i = #1}^{#2}}
\newcommand{\datasetMMLifted}[2]{(\tilde\shearmap_i,\, \convmap_i)_{i = #1}^{#2}}

\newcommand{\trainingsetMM}{\datasetMM{1}{\sizeTrainingset}}
\newcommand{\trainingsetMMLifted}{\datasetMMLifted{1}{\sizeTrainingset}}

\newcommand{\calibrationsetMMLifted}{\datasetMMLifted{\sizeTrainingset + 1}{\sizeTrainingset + \sizeCalibrationset}}

\newcommand{\datasetMMLiftedRand}[2]{(\tilde\shearmapRand_i,\, \convmapRand_i)_{i = #1}^{#2}}

\newcommand{\calibrationsetMMLiftedRand}{\datasetMMLiftedRand{\sizeTrainingset + 1}{\sizeTrainingset + \sizeCalibrationset}}

\newcommand{\priorpdf}[2]{\pdf_{#1}\!\left(#2\right)}
\newcommand{\condpriorpdf}[2]{\pdf_{#1}\!\left(#2\right)}

\newcommand{\cond}[2]{#1 \,|\, #2}

\newcommand{\kerBA}{\ker\tilde\genericop\tilde\convToShear}

\newcommand{\orthcomplBA}{(\kerBA)^\perp}

%% file: inputs/flowchart.tex
\begin{figure}
    \centering\small
    \begin{tikzpicture}

        \node (init) {$\convmapEstimate^{(0)}$};
        
        \node[draw, anchor=center, yshift=-0.7cm] at (init.south) (forward1) {Forward};
        \node[draw, anchor=west, xshift=0.5cm, yshift=-0.7cm] at (forward1.east) (backward1) {Backward -- $\deepmodelTrained$};
        
        \node[draw, anchor=center, yshift=-1.4cm] at (forward1.center) (forward2) {Forward};
        \node[draw, anchor=center, yshift=-1.4cm] at (backward1.center) (backward2) {Backward -- $\deepmodelTrained$};
        
        \node[draw, anchor=center, yshift=-1.4cm] at (forward2.center) (forward3) {Forward};
        \node[draw, anchor=center, yshift=-1.4cm] at (backward2.center) (backward3) {Backward -- $\deepmodelbisTrained$};
        
        \node[anchor=east, xshift=-0.5cm] at (forward1.west) (shearmap1) {$\tilde\shearmap$};
        \node[anchor=east, xshift=-0.5cm] at (forward2.west) (shearmap2) {$\tilde\shearmap$};
        \node[anchor=east, xshift=-0.5cm] at (forward3.west) (shearmap3) {$\tilde\shearmap$};
        
        \node[anchor=west, xshift=0.5cm] at (backward2.east) (postmean) {$\convmapEstimate$ (point estimate)};
        \node[anchor=west, xshift=0.5cm] at (backward3.east) (postvar) {$\varmapEstimate$ (variance estimate)};

        \draw[->] (init) -- (forward1);
        \draw[->] (forward1.south east) -- (backward1.north west);
        \draw[->] (backward1.south west) -- (forward2.north east);
        \draw[->] (forward2.south east) -- (backward2.north west);
        \draw[->] (backward2.south west) -- (forward3.north east);
        \draw[->] (forward3.south east) -- (backward3.north west);
        
        \draw[->] (shearmap1) -- (forward1);
        \draw[->] (shearmap2) -- (forward2);
        \draw[->] (shearmap3) -- (forward3);
        
        \draw[->] (backward2) -- (postmean);
        \draw[->] (backward3) -- (postvar);
        
    \end{tikzpicture}
    \caption{Schematic representation of the PnPMass algorithm. Two iterations are shown for the point estimate (for illustration purposes only; more iterations are typically required in practice), followed by one additional iteration for the variance.}
    \label{fig:pnpmass}
\end{figure}

%% file: bibtex/refs.bib
@article{AbdarReviewUncertainty2021,
  title = {A Review of Uncertainty Quantification in Deep Learning: {{Techniques}}, Applications and Challenges},
  shorttitle = {A Review of Uncertainty Quantification in Deep Learning},
  author = {Abdar, Moloud and Pourpanah, Farhad and Hussain, Sadiq and Rezazadegan, Dana and Liu, Li and Ghavamzadeh, Mohammad and Fieguth, Paul and Cao, Xiaochun and Khosravi, Abbas and Acharya, U. Rajendra and Makarenkov, Vladimir and Nahavandi, Saeid},
  year = {2021},
  month = dec,
  njournal = {Information Fusion},
  journal = {Inf.\@ Fusion},
  volume = {76},
  pages = {243--297},
  issn = {1566-2535},
  doi = {10.1016/j.inffus.2021.05.008},
  urldate = {2023-10-24}
}

@inproceedings{Angelopoulos2022c,
  title = {Image-to-{{Image Regression}} with {{Distribution-Free Uncertainty Quantification}} and {{Applications}} in {{Imaging}}},
  booktitle = {Proceedings of the 39th {{International Conference}} on {{Machine Learning}}},
  author = {Angelopoulos, Anastasios N. and Kohli, Amit Pal and Bates, Stephen and Jordan, Michael and Malik, Jitendra and Alshaabi, Thayer and Upadhyayula, Srigokul and Romano, Yaniv},
  year = {2022},
  month = jun,
  pages = {717--730},
  publisher = {PMLR},
  issn = {2640-3498},
  urldate = {2023-11-13},
  langid = {english}
}

@article{AoyamaDenoisingWeak2025,
  title = {Denoising Weak Lensing Mass Maps with Diffusion Model: Systematic Comparison with Generative Adversarial Network},
  shorttitle = {Denoising Weak Lensing Mass Maps with Diffusion Model},
  author = {Aoyama, Shohei D. and Osato, Ken and Shirasaki, Masato},
  year = {2025},
  month = may,
  journal = {PASJ, submitted},
  eprint = {2505.00345},
  primaryclass = {astro-ph},
  doi = {10.48550/arXiv.2505.00345},
  urldate = {2025-06-02},
  archiveprefix = {arXiv}
}

@article{Bobin2012,
	title = {{CMB} {Map} {Restoration}},
	volume = {2012},
	copyright = {Copyright © 2012 J. Bobin et al.},
	issn = {1687-7977},
	url = {https://onlinelibrary.wiley.com/doi/abs/10.1155/2012/703217},
	doi = {10.1155/2012/703217},
	language = {en},
	number = {1},
	urldate = {2026-03-26},
	njournal = {Advances in Astronomy},
  journal = {Adv.\@ Astron.},
	author = {Bobin, J. and Starck, J.-L. and Sureau, F. and Fadili, J.},
	year = {2012},
	nnote = {\_eprint: https://onlinelibrary.wiley.com/doi/pdf/10.1155/2012/703217},
	pages = {703217},
}

@article{Ebner2024,
  title = {Plug-and-{{Play Image Reconstruction Is}} a {{Convergent Regularization Method}}},
  author = {Ebner, Andrea and Haltmeier, Markus},
  year = {2024},
  njournal = {IEEE Transactions on Image Processing},
  journal = {IEEE Trans.\@ Image Process.},
  volume = {33},
  pages = {1476--1486},
  issn = {1941-0042},
  doi = {10.1109/TIP.2024.3361218},
  urldate = {2024-11-22}
}

@inproceedings{Ekmekci2021,
  title = {What {{Does Your Computational Imaging Algorithm Not Know}}?: {{A Plug-and-Play Model Quantifying Model Uncertainty}}},
  shorttitle = {What {{Does Your Computational Imaging Algorithm Not Know}}?},
  booktitle = {Proceedings of the {{IEEE}}/{{CVF International Conference}} on {{Computer Vision}}},
  author = {Ekmekci, Canberk and Cetin, Mujdat},
  year = {2021},
  pages = {4018--4027},
  urldate = {2024-12-06},
  langid = {english}
}

@inproceedings{Fan2022,
  title = {{{SUNet}}: {{Swin Transformer UNet}} for {{Image Denoising}}},
  shorttitle = {{{SUNet}}},
  booktitle = {2022 {{IEEE International Symposium}} on {{Circuits}} and {{Systems}} ({{ISCAS}})},
  author = {Fan, Chi-Mao and Liu, Tsung-Jung and Liu, Kuan-Hsien},
  year = {2022},
  month = may,
  pages = {2333--2337},
  issn = {2158-1525},
  doi = {10.1109/ISCAS48785.2022.9937486},
  urldate = {2025-02-10}
}

@article{GibbsConformalPrediction2025,
	title = {Conformal prediction with conditional guarantees},
	volume = {87},
	issn = {1369-7412},
	url = {https://doi.org/10.1093/jrsssb/qkaf008},
	doi = {10.1093/jrsssb/qkaf008},
	number = {4},
	urldate = {2026-03-26},
	njournal = {Journal of the Royal Statistical Society Series B: Statistical Methodology},
  journal = {J.\@ R.\@ Stat.\@ Soc.\@ Ser.\@ B: Stat.\@ Methodol.},
	author = {Gibbs, Isaac and Cherian, John J and Candès, Emmanuel J},
	month = sep,
	year = {2025},
	pages = {1100--1126},
}

@inproceedings{Hurault2022,
  title = {Proximal {{Denoiser}} for {{Convergent Plug-and-Play Optimization}} with {{Nonconvex Regularization}}},
  booktitle = {Proceedings of the 39th {{International Conference}} on {{Machine Learning}}},
  author = {Hurault, Samuel and Leclaire, Arthur and Papadakis, Nicolas},
  year = {2022},
  month = jun,
  pages = {9483--9505},
  publisher = {PMLR},
  issn = {2640-3498},
  urldate = {2024-12-04},
  langid = {english}
}

@article{Jeffrey2020,
  title = {Deep Learning Dark Matter Map Reconstructions from {{DES SV}} Weak Lensing Data},
  author = {Jeffrey, Niall and Lanusse, Fran{\c c}ois and Lahav, Ofer and Starck, Jean-Luc},
  year = {2020},
  month = mar,
  njournal = {Monthly Notices of the Royal Astronomical Society},
  journal = {MNRAS},
  volume = {492},
  number = {4},
  pages = {5023--5029},
  issn = {0035-8711},
  doi = {10.1093/mnras/staa127},
  urldate = {2023-05-29}
}

@inproceedings{Jeffrey2020a,
  title = {Solving High-Dimensional Parameter Inference: Marginal Posterior Densities \& {{Moment Networks}}},
  shorttitle = {Solving High-Dimensional Parameter Inference},
  booktitle = {Third {{Workshop}} on {{Machine Learning}} and the {{Physical Sciences}} ({{NeurIPS}} 2020)},
  author = {Jeffrey, Niall and Wandelt, Benjamin D.},
  year = {2020},
  month = nov,
  eprint = {2011.05991},
  primaryclass = {astro-ph, stat},
  doi = {10.48550/arXiv.2011.05991},
  urldate = {2024-05-23},
  archiveprefix = {arXiv}
}

@article{Kaiser1993,
  title = {Mapping the Dark Matter with Weak Gravitational Lensing},
  author = {Kaiser, Nick and Squires, Gordon},
  year = {1993},
  njournal = {Astrophysical Journal},
  journal = {ApJ},
  volume = {404},
  number = {2},
  pages = {441--450},
  issn = {0004-637X},
  urldate = {2023-11-28}
}

@article{Kamilov2023,
  title = {Plug-and-{{Play Methods}} for {{Integrating Physical}} and {{Learned Models}} in {{Computational Imaging}}: {{Theory}}, Algorithms, and Applications},
  shorttitle = {Plug-and-{{Play Methods}} for {{Integrating Physical}} and {{Learned Models}} in {{Computational Imaging}}},
  author = {Kamilov, Ulugbek S. and Bouman, Charles A. and Buzzard, Gregery T. and Wohlberg, Brendt},
  year = {2023},
  month = jan,
  njournal = {IEEE Signal Processing Magazine},
  journal = {IEEE Signal Proc.\@ Mag.},
  volume = {40},
  number = {1},
  pages = {85--97},
  issn = {1558-0792},
  doi = {10.1109/MSP.2022.3199595},
  urldate = {2023-10-10}
}

@article{Kilbinger2015,
  title = {Cosmology with Cosmic Shear Observations: A Review},
  shorttitle = {Cosmology with Cosmic Shear Observations},
  author = {Kilbinger, Martin},
  year = {2015},
  month = jul,
  njournal = {Reports on Progress in Physics},
  journal = {Rep.\@ Prog.\@ Phys.},
  volume = {78},
  number = {8},
  pages = {086901},
  publisher = {IOP Publishing},
  issn = {0034-4885},
  doi = {10.1088/0034-4885/78/8/086901},
  urldate = {2023-06-11},
  langid = {english}
}

@article{Lanusse2016,
  title = {High Resolution Weak Lensing Mass Mapping Combining Shear and Flexion},
  author = {Lanusse, F. and Starck, J.-L. and Leonard, A. and Pires, S.},
  year = {2016},
  month = jul,
  njournal = {Astronomy \& Astrophysics},
  journal = {A\&A},
  volume = {591},
  pages = {A2},
  publisher = {EDP Sciences},
  issn = {0004-6361, 1432-0746},
  doi = {10.1051/0004-6361/201628278},
  urldate = {2024-04-16},
  copyright = {{\copyright} ESO, 2016},
  langid = {english}
}

@article{LaumontBayesianImaging2022,
  title = {Bayesian {{Imaging Using Plug}} \& {{Play Priors}}: {{When Langevin Meets Tweedie}}},
  shorttitle = {Bayesian {{Imaging Using Plug}} \& {{Play Priors}}},
  author = {Laumont, R{\'e}mi and Bortoli, Valentin De and Almansa, Andr{\'e}s and Delon, Julie and Durmus, Alain and Pereyra, Marcelo},
  year = {2022},
  month = jun,
  njournal = {SIAM Journal on Imaging Sciences},
  journal = {SIAM J.\@ Imaging Sci.},
  volume = {15},
  number = {2},
  pages = {701--737},
  publisher = {{Society for Industrial and Applied Mathematics}},
  doi = {10.1137/21M1406349},
  urldate = {2023-10-06}
}

@article{LetermeDistributionfreeUncertainty2025,
  title = {Distribution-Free Uncertainty Quantification for Inverse Problems: {{Application}} to Weak Lensing Mass Mapping},
  shorttitle = {Distribution-Free Uncertainty Quantification for Inverse Problems},
  author = {Leterme, H. and Fadili, J. and Starck, J.-L.},
  year = {2025},
  month = feb,
  njournal = {Astronomy \& Astrophysics},
  journal = {A\&A},
  volume = {694},
  pages = {A267},
  publisher = {EDP Sciences},
  issn = {0004-6361, 1432-0746},
  doi = {10.1051/0004-6361/202451756},
  urldate = {2025-03-05},
  copyright = {{\copyright} The Authors 2025},
  langid = {english}
}

@article{Mobasher2007,
  title = {Photometric {{Redshifts}} of {{Galaxies}} in {{COSMOS}}*},
  author = {Mobasher, B. and Capak, P. and Scoville, N. Z. and Dahlen, T. and Salvato, M. and Aussel, H. and Thompson, D. J. and Feldmann, R. and Tasca, L. and Lefevre, O. and Lilly, S. and Carollo, C. M. and Kartaltepe, J. S. and McCracken, H. and Mould, J. and Renzini, A. and Sanders, D. B. and Shopbell, P. L. and Taniguchi, Y. and Ajiki, M. and Shioya, Y. and Contini, T. and Giavalisco, M. and Ilbert, O. and Iovino, A. and Brun, V. Le and Mainieri, V. and Mignoli, M. and Scodeggio, M.},
  year = {2007},
  month = sep,
  njournal = {The Astrophysical Journal Supplement Series},
  journal = {ApJS},
  volume = {172},
  number = {1},
  pages = {117},
  issn = {0067-0049},
  doi = {10.1086/516590},
  urldate = {2024-07-12},
  langid = {english}
}

@article{Osato2021,
  title = {{{$\kappa$TNG}}: Effect of Baryonic Processes on Weak Lensing with {{IllustrisTNG}} Simulations},
  shorttitle = {{{$\kappa$TNG}}},
  author = {Osato, Ken and Liu, Jia and Haiman, Zolt{\'a}n},
  year = {2021},
  month = apr,
  njournal = {Monthly Notices of the Royal Astronomical Society},
  journal = {MNRAS},
  volume = {502},
  number = {4},
  pages = {5593--5602},
  issn = {0035-8711},
  doi = {10.1093/mnras/stab395},
  urldate = {2024-05-02}
}

@article{Pesquet2021,
  title = {Learning {{Maximally Monotone Operators}} for {{Image Recovery}}},
  author = {Pesquet, Jean-Christophe and Repetti, Audrey and Terris, Matthieu and Wiaux, Yves},
  year = {2021},
  month = jan,
  njournal = {SIAM Journal on Imaging Sciences},
  journal = {SIAM J.\@ Imaging Sci.},
  volume = {14},
  number = {3},
  pages = {1206--1237},
  publisher = {{Society for Industrial and Applied Mathematics}},
  doi = {10.1137/20M1387961},
  urldate = {2024-05-23}
}

@inproceedings{PostelsSamplingFreeEpistemic2019,
  title = {Sampling-{{Free Epistemic Uncertainty Estimation Using Approximated Variance Propagation}}},
  booktitle = {Proceedings of the {{IEEE}}/{{CVF International Conference}} on {{Computer Vision}}},
  author = {Postels, Janis and Ferroni, Francesco and Coskun, Huseyin and Navab, Nassir and Tombari, Federico},
  year = {2019},
  pages = {2931--2940},
  urldate = {2025-06-17}
}

@article{Remy2023,
  title = {Probabilistic Mass-Mapping with Neural Score Estimation},
  author = {Remy, B. and Lanusse, F. and Jeffrey, N. and Liu, J. and Starck, J.-L. and Osato, K. and Schrabback, T.},
  year = {2023},
  month = apr,
  njournal = {Astronomy \& Astrophysics},
  journal = {A\&A},
  volume = {672},
  pages = {A51},
  publisher = {EDP Sciences},
  issn = {0004-6361, 1432-0746},
  doi = {10.1051/0004-6361/202243054},
  urldate = {2023-10-10},
  copyright = {{\copyright} The Authors 2023},
  langid = {english}
}

@inproceedings{Romano2019,
  title = {Conformalized {{Quantile Regression}}},
  booktitle = {Advances in {{Neural Information Processing Systems}}},
  author = {Romano, Yaniv and Patterson, Evan and Candes, Emmanuel},
  year = {2019},
  volume = {32},
  publisher = {Curran Associates, Inc.},
  urldate = {2023-05-29}
}

@inproceedings{Ryu2019,
  title = {Plug-and-{{Play Methods Provably Converge}} with {{Properly Trained Denoisers}}},
  booktitle = {Proceedings of the 36th {{International Conference}} on {{Machine Learning}}},
  author = {Ryu, Ernest and Liu, Jialin and Wang, Sicheng and Chen, Xiaohan and Wang, Zhangyang and Yin, Wotao},
  year = {2019},
  month = may,
  pages = {5546--5557},
  publisher = {PMLR},
  issn = {2640-3498},
  urldate = {2024-11-22},
  langid = {english}
}

@article{Schrabback2010,
  title = {Evidence of the Accelerated Expansion of the {{Universe}} from Weak Lensing Tomography with {{COSMOS}}},
  author = {Schrabback, T. and Hartlap, J. and Joachimi, B. and Kilbinger, M. and Simon, P. and Benabed, K. and Brada{\v c}, M. and Eifler, T. and Erben, T. and Fassnacht, C. D. and High, F. William and Hilbert, S. and Hildebrandt, H. and Hoekstra, H. and Kuijken, K. and Marshall, P. J. and Mellier, Y. and Morganson, E. and Schneider, P. and Semboloni, E. and Waerbeke, L. Van and Velander, M.},
  year = {2010},
  month = jun,
  njournal = {Astronomy \& Astrophysics},
  journal = {A\&A},
  volume = {516},
  pages = {A63},
  publisher = {EDP Sciences},
  issn = {0004-6361, 1432-0746},
  doi = {10.1051/0004-6361/200913577},
  urldate = {2023-12-21},
  copyright = {{\copyright} ESO, 2010},
  langid = {english}
}

@article{Scoville2007,
  title = {{{COSMOS}}: {{Hubble Space Telescope Observations}}*},
  shorttitle = {{{COSMOS}}},
  author = {Scoville, N. and Abraham, R. G. and Aussel, H. and Barnes, J. E. and Benson, A. and Blain, A. W. and Calzetti, D. and Comastri, A. and Capak, P. and Carilli, C. and Carlstrom, J. E. and Carollo, C. M. and Colbert, J. and Daddi, E. and Ellis, R. S. and Elvis, M. and Ewald, S. P. and Fall, M. and Franceschini, A. and Giavalisco, M. and Green, W. and Griffiths, R. E. and Guzzo, L. and Hasinger, G. and Impey, C. and Kneib, J.-P. and Koda, J. and Koekemoer, A. and Lefevre, O. and Lilly, S. and Liu, C. T. and McCracken, H. J. and Massey, R. and Mellier, Y. and Miyazaki, S. and Mobasher, B. and Mould, J. and Norman, C. and Refregier, A. and Renzini, A. and Rhodes, J. and Rich, M. and Sanders, D. B. and Schiminovich, D. and Schinnerer, E. and Scodeggio, M. and Sheth, K. and Shopbell, P. L. and Taniguchi, Y. and Tyson, N. D. and Urry, C. M. and Waerbeke, L. Van and Vettolani, P. and White, S. D. M. and Yan, L.},
  year = {2007},
  month = sep,
  njournal = {The Astrophysical Journal Supplement Series},
  journal = {ApJS},
  volume = {172},
  number = {1},
  pages = {38},
  issn = {0067-0049},
  doi = {10.1086/516580},
  urldate = {2024-07-12},
  langid = {english}
}

@inproceedings{Shi2021,
  title = {Fast {{Uncertainty Quantification}} for {{Deep Object Pose Estimation}}},
  booktitle = {2021 {{IEEE International Conference}} on {{Robotics}} and {{Automation}} ({{ICRA}})},
  author = {Shi, Guanya and Zhu, Yifeng and Tremblay, Jonathan and Birchfield, Stan and Ramos, Fabio and Anandkumar, Animashree and Zhu, Yuke},
  year = {2021},
  month = may,
  pages = {5200--5207},
  issn = {2577-087X},
  doi = {10.1109/ICRA48506.2021.9561483},
  urldate = {2024-12-06}
}

@article{Shirasaki2019,
  title = {Denoising Weak Lensing Mass Maps with Deep Learning},
  author = {Shirasaki, Masato and Yoshida, Naoki and Ikeda, Shiro},
  year = {2019},
  month = aug,
  njournal = {Physical Review D},
  journal = {Phys.\@ Rev.\@ D},
  volume = {100},
  number = {4},
  pages = {043527},
  publisher = {American Physical Society},
  doi = {10.1103/PhysRevD.100.043527},
  urldate = {2023-05-29}
}

@article{Shirasaki2021,
  title = {Noise Reduction for Weak Lensing Mass Mapping: An Application of Generative Adversarial Networks to {{Subaru Hyper Suprime-Cam}} First-Year Data},
  shorttitle = {Noise Reduction for Weak Lensing Mass Mapping},
  author = {Shirasaki, Masato and Moriwaki, Kana and Oogi, Taira and Yoshida, Naoki and Ikeda, Shiro and Nishimichi, Takahiro},
  year = {2021},
  month = jun,
  njournal = {Monthly Notices of the Royal Astronomical Society},
  journal = {MNRAS},
  volume = {504},
  number = {2},
  pages = {1825--1839},
  issn = {0035-8711},
  doi = {10.1093/mnras/stab982},
  urldate = {2024-04-17}
}

@article{Starck2013,
  title = {Sparsity and the {{Bayesian}} Perspective},
  author = {Starck, J.-L. and Donoho, D. L. and Fadili, M. J. and Rassat, A.},
  year = {2013},
  month = apr,
  njournal = {Astronomy \& Astrophysics},
  journal = {A\&A},
  volume = {552},
  pages = {A133},
  publisher = {EDP Sciences},
  issn = {0004-6361, 1432-0746},
  doi = {10.1051/0004-6361/201321257},
  urldate = {2024-04-24},
  copyright = {{\copyright} ESO, 2013},
  langid = {english}
}

@book{Starck2015,
  title = {Sparse {{Image}} and {{Signal Processing}}: {{Wavelets}} and {{Related Geometric Multiscale Analysis}}},
  shorttitle = {Sparse {{Image}} and {{Signal Processing}}},
  author = {Starck, Jean-Luc and Murtagh, Fionn and Fadili, Jalal},
  year = {2015},
  month = oct,
  publisher = {Cambridge University Press},
  googlebooks = {wUy2CgAAQBAJ},
  isbn = {978-1-107-08806-1},
  langid = {english}
}

@article{Starck2021,
  title = {Weak-Lensing Mass Reconstruction Using Sparsity and a {{Gaussian}} Random Field},
  author = {Starck, J.-L. and Themelis, K. E. and Jeffrey, N. and Peel, A. and Lanusse, F.},
  year = {2021},
  month = may,
  njournal = {Astronomy \& Astrophysics},
  journal = {A\&A},
  volume = {649},
  pages = {A99},
  publisher = {EDP Sciences},
  issn = {0004-6361, 1432-0746},
  doi = {10.1051/0004-6361/202039451},
  urldate = {2023-10-06},
  copyright = {{\copyright} J.-L. Starck et al. 2021},
  langid = {english}
}

@inproceedings{Terris2020,
  title = {Building {{Firmly Nonexpansive Convolutional Neural Networks}}},
  booktitle = {{{IEEE International Conference}} on {{Acoustics}}, {{Speech}} and {{Signal Processing}} ({{ICASSP}})},
  nbooktitle = {{{ICASSP}} 2020 - 2020 {{IEEE International Conference}} on {{Acoustics}}, {{Speech}} and {{Signal Processing}} ({{ICASSP}})},
  author = {Terris, Matthieu and Repetti, Audrey and Pesquet, Jean-Christophe and Wiaux, Yves},
  year = {2020},
  month = may,
  npages = {8658--8662},
  issn = {2379-190X},
  doi = {10.1109/ICASSP40776.2020.9054731},
  urldate = {2024-12-04}
}

@article{TerrisAIRIPlugandplayAlgorithm2025,
  title = {The {{AIRI}} Plug-and-Play Algorithm for Image Reconstruction in Radio-Interferometry: Variations and Robustness},
  shorttitle = {The {{AIRI}} Plug-and-Play Algorithm for Image Reconstruction in Radio-Interferometry},
  author = {Terris, Matthieu and Tang, Chao and Jackson, Adrian and Wiaux, Yves},
  year = {2025},
  month = feb,
  njournal = {Monthly Notices of the Royal Astronomical Society},
  journal = {MNRAS},
  volume = {537},
  number = {2},
  pages = {1608--1619},
  issn = {0035-8711},
  doi = {10.1093/mnras/staf022},
  urldate = {2025-10-10}
}

@article{TersenovImpactWeaklensingMassmapping2025,
  title = {Impact of Weak-Lensing Mass-Mapping Algorithms on Cosmology Inference},
  author = {Tersenov, Andreas and Baumont, Lucie and Starck, Jean-Luc and Kilbinger, Martin},
  year = {2025},
  month = jun,
  njournal = {Astronomy \& Astrophysics},
  journal = {A\&A},
  volume = {698},
  pages = {A25},
  publisher = {EDP Sciences},
  issn = {0004-6361, 1432-0746},
  doi = {10.1051/0004-6361/202553707},
  urldate = {2025-10-10},
  copyright = {{\copyright} The Authors 2025},
  langid = {english}
}

@inproceedings{Venkatakrishnan2013,
  title = {Plug-and-{{Play}} Priors for Model Based Reconstruction},
  booktitle = {2013 {{IEEE Global Conference}} on {{Signal}} and {{Information Processing}}},
  author = {Venkatakrishnan, Singanallur V. and Bouman, Charles A. and Wohlberg, Brendt},
  year = {2013},
  month = dec,
  pages = {945--948},
  doi = {10.1109/GlobalSIP.2013.6737048},
  urldate = {2023-11-09}
}

@article{WhitneyGenerativeModellingMassmapping2025,
  title = {Generative Modelling for Mass-Mapping with Fast Uncertainty Quantification},
  author = {Whitney, Jessica J and Liaudat, Tob{\'i}as I and Price, Matthew A and Mars, Matthijs and McEwen, Jason D},
  year = {2025},
  month = sep,
  njournal = {Monthly Notices of the Royal Astronomical Society},
  journal = {MNRAS},
  volume = {542},
  number = {3},
  pages = {2464--2479},
  issn = {0035-8711},
  doi = {10.1093/mnras/staf1356},
  urldate = {2025-10-10}
}

@article{ZhangPlugandPlayImage2022,
  title = {Plug-and-{{Play Image Restoration With Deep Denoiser Prior}}},
  author = {Zhang, Kai and Li, Yawei and Zuo, Wangmeng and Zhang, Lei and Van Gool, Luc and Timofte, Radu},
  year = {2022},
  month = oct,
  njournal = {IEEE Transactions on Pattern Analysis and Machine Intelligence},
  journal = {IEEE Trans.\@ Pattern Anal.\@ Mach.\@ Intell.},
  volume = {44},
  number = {10},
  pages = {6360--6376},
  issn = {1939-3539},
  doi = {10.1109/TPAMI.2021.3088914},
  urldate = {2025-05-15}
}

@article{FinoguenovXmmNewton2007,
	title = {The {XMM}-{Newton} {Wide}-{Field} {Survey} in the {COSMOS} {Field}: {Statistical} {Properties} of {Clusters} of {Galaxies}},
	volume = {172},
	issn = {0067-0049},
	shorttitle = {The {XMM}-{Newton} {Wide}-{Field} {Survey} in the {COSMOS} {Field}},
	url = {https://doi.org/10.1086/516577},
	doi = {10.1086/516577},
	language = {en},
	number = {1},
	urldate = {2026-03-06},
	njournal = {The Astrophysical Journal Supplement Series},
  journal = {ApJS},
	author = {Finoguenov, A. and Guzzo, L. and Hasinger, G. and Scoville, N. Z. and Aussel, H. and Böhringer, H. and Brusa, M. and Capak, P. and Cappelluti, N. and Comastri, A. and Giodini, S. and Griffiths, R. E. and Impey, C. and Koekemoer, A. M. and Kneib, J.-P. and Leauthaud, A. and Le Fèvre, O. and Lilly, S. and Mainieri, V. and Massey, R. and McCracken, H. J. and Mobasher, B. and Murayama, T. and Peacock, J. A. and Sakelliou, I. and Schinnerer, E. and Silverman, J. D. and Smolčić, V. and Taniguchi, Y. and Tasca, L. and Taylor, J. E. and Trump, J. R. and Zamorani, G.},
	month = sep,
	year = {2007},
	pages = {182},
}

@article{TachellaDeepInverse2025,
    title = {DeepInverse: A Python package for solving imaging inverse problems with deep learning},
    njournal = {Journal of Open Source Software},
    journal = {J.\@ Open Source Software},
    doi = {10.21105/joss.08923},
    url = {https://doi.org/10.21105/joss.08923},
    year = {2025},
    publisher = {The Open Journal},
    volume = {10},
    number = {115},
    pages = {8923},
    author = {Tachella, Julián and Terris, Matthieu and Hurault, Samuel and Wang, Andrew and Davy, Leo and Scanvic, Jérémy and Sechaud, Victor and Vo, Romain and Moreau, Thomas and Davies, Thomas and Chen, Dongdong and Laurent, Nils and Monroy, Brayan and Dong, Jonathan and Hu, Zhiyuan and Nguyen, Minh-Hai and Sarron, Florian and Weiss, Pierre and Escande, Paul and Massias, Mathurin and Modrzyk, Thibaut and Levac, Brett and Liaudat, Tobías I. and Song, Maxime and Hertrich, Johannes and Neumayer, Sebastian and Schramm, Georg},
}
